\documentclass[ejs]{imsart}

\doi{10.1214/154957804100000000}
\pubyear{0000}
\volume{0}
\firstpage{1}
\lastpage{0}

\usepackage[round]{natbib}
\usepackage[T1]{fontenc}
\usepackage[utf8]{inputenc}
\usepackage[raggedright]{titlesec}
\usepackage[all]{onlyamsmath}
\usepackage{amsmath,amssymb,amsthm,bm,mathtools,xspace,booktabs,url,graphicx,etoolbox,multirow,subfigure,float,soul,siunitx,hyperref}
\usepackage{color}

\numberwithin{equation}{section} 
\newtheorem{theorem}{Theorem}[section]

\newtheorem{proposition}[theorem]{Proposition}
\newtheorem{definition}[theorem]{Definition}

\DeclareMathOperator\E{E}
\DeclareMathOperator\PP{P}

\newcommand{\dd}{\,\mathrm d}

\newcommand{\R}{\mathbb{R}}

\newcommand{\bx}{\mathbf{x}}

\newcommand{\bX}{\mathbf{X}}
\newcommand{\bY}{\mathbf{Y}}
\newcommand{\bZ}{\mathbf{Z}}

\newcommand{\PS}{\mathrm{PS}}
\newcommand{\DPP}{\mathrm{DPP}}
\newcommand{\sPS}{{\textsc{\scriptsize ps}}}
\newcommand{\sDPP}{{\textsc{\scriptsize dpp}}}
\newcommand{\lPS}{\lambda_{\textsc{\scriptsize ps}}}
\newcommand{\lDPP}{\lambda_{\textsc{\scriptsize   dpp}}}

\graphicspath{{.}{fig/}}

\begin{document}

\begin{frontmatter}

\title{Approximation intensity for pairwise interaction Gibbs point processes using determinantal point processes}  
\runtitle{Intensity approximation using DPP}

\begin{aug}

\author{\fnms{Jean-Fran\c cois} \snm{Coeurjolly}\thanksref{a,b}\ead[label=e1]{coeurjolly.jean-francois@uqam.ca}\ead[label=e2]{jean-francois.coeurjolly@univ-grenoble-alpes.fr}}
\and
\author{\fnms{Frédéric} \snm{Lavancier}\thanksref{c}\ead[label=e3]{frederic.lavancier@univ-nantes.fr}}

\address[a]{Department of Mathematics, Universit\'e du Qu\'ebec \`a Montr\'eal (UQAM), Canada \\ \printead{e1}}
\address[b]{Laboratory Jean Kuntzmann, Universit\'e Grenoble Alpes, CNRS, France
 \\ \printead{e2}}
\address[c]{Laboratoire de Math\'ematiques Jean Leray - Universit\'e de Nantes, France \\ \printead{e3}}
\runauthor{Coeurjolly and Lavancier}

\end{aug}

\begin{abstract}
The intensity of a Gibbs point process is usually an intractable function of the model parameters. For repulsive pairwise interaction point processes, this intensity can be expressed as the Laplace transform of some particular function. 
\citet{baddeley:nair:12} developped the Poisson-saddlepoint approximation which consists, for basic models, in calculating this Laplace transform with respect to a homogeneous Poisson point process. In this paper, we develop an approximation which  consists in calculating the same Laplace transform with respect to a specific determinantal point process. This new approximation is efficiently implemented and turns out to be more accurate than the Poisson-saddlepoint approximation, as demonstrated by some  numerical examples. 
\end{abstract}



\received{\smonth{1} \syear{0000}}

\tableofcontents

\end{frontmatter}

\section{Introduction}\label{sec:intro}

Due to their simple interpretation, Gibbs point processes and in particular pairwise interaction point processes play a central role in the analysis of spatial point patterns (see \citet{lieshout:00,moeller:waagepetersen:04,baddeley:rubak:turner:15}). In a nutshell, such models (in the homogeneous case) are defined in a bounded domain by a density with respect to the unit rate Poisson point process which takes the form
\[
  f(\bx) \propto \beta^{|\bx|}\prod_{u\in \bx}g(v-u), 
\]
where $\bx$ is a finite configuration of points, where $\beta>0$ represents the activity parameter, $|\bx|$ is the number of elements of $\bx$ and where $g:\R^d \to \R^+$ is the pairwise interaction function. 

However, many important theoretical properties of these models are in general intractable, like for instance the simplest one, the intensity $\lambda\in \R^+$, representing  the mean number of points per unit volume. It is known (see e.g. Section~\ref{sec:gibbs}) that 
\[
  \lambda = \beta \E\left( \prod_{u\in \bx} g(u)\right).
\]
Such an expectation is in general intractable. As clearly outlined by \citet{baddeley:nair:12}, this intractability constitutes a severe drawback. For example, simulating a Gibbs point process with a prescribed value of $\lambda$ cannot be done beforehand even for simple models such as Strauss models.
\citet{baddeley:nair:12} suggest to evaluate the expectation with respect to a homogeneous Poisson point process with intensity $\lambda$. This results in the Poisson-saddlepoint approximation, denoted by $\lPS$, obtained as the solution of
\[\log\lPS = \log\beta -\lPS  \, G\]
where $G = \int_{\R^d}(1-g(u))\dd u$ (provided this integral is finite). 

The general idea of the present paper is to evaluate the same expectaction with respect to a determinantal point process (with intensity $\lambda$). Determinantal point processes (DPP), see e.g.~\citet{lavancier:moeller:rubak:15}, are  a class of repulsive models which is more tractable than Gibbs models. For example all moments are explicit. If $g\leq 1$ and has a finite range $R>0$, our approximation denoted by $\lDPP$ is the solution of
\[
  \log\lDPP  = \log\beta + (1+\lDPP G/\kappa )  \log  \left(1- \frac{\lDPP G} {1+\lDPP G/\kappa}\right).
\]
where 
\[\kappa = \max \left(\frac{|B(0,\delta)|}{\int (1-g)^2} , \frac{\int (1-g)^2}{|B(0,R)|}\right),
\]
$|A|$ denotes the volume of some bounded domain $A\subset \R^d$, $B(0,\rho)$ is the Euclidean ball centered at 0 with radius $\rho$ and $\delta\geq 0$ is some possible hard-core distance.

Both approximations $\lDPP$ and $\lPS$  can be obtained very quickly with a unit-root search algorithm. Figure~\ref{fig:straussintro} reports $\lDPP$ and $\lPS$ as well as the true intensity $\lambda$ (obtained by Monte-Carlo simulations) for Strauss models in terms of the interaction parameter $\gamma_1 \in [0,1]$. This setting is considered by \citet{baddeley:nair:12}.
The DPP approximation outperforms the Poisson-saddlepoint approximation especially when $\gamma_1$ is close to zero, i.e. for very repulsive point processes.
More numerical illustrations are displayed in Section~\ref{sec:simu}. 

\begin{figure}[h]
\subfigure[Strauss: $R=0.05$, $\beta=100$]{\includegraphics[scale=.4]{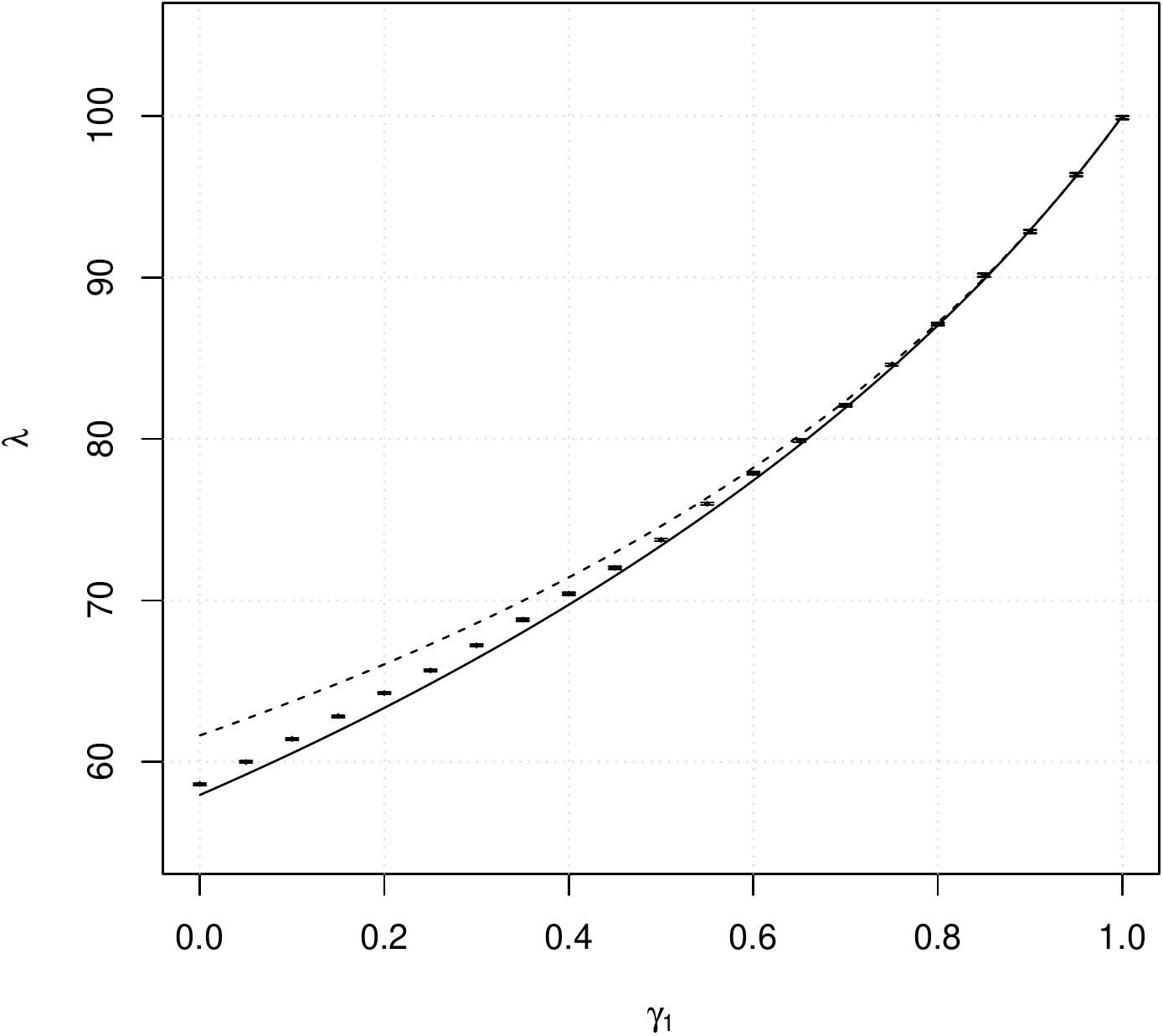}}
\subfigure[Strauss: $R=0.1$, $\beta=100$]{\includegraphics[scale=.4]{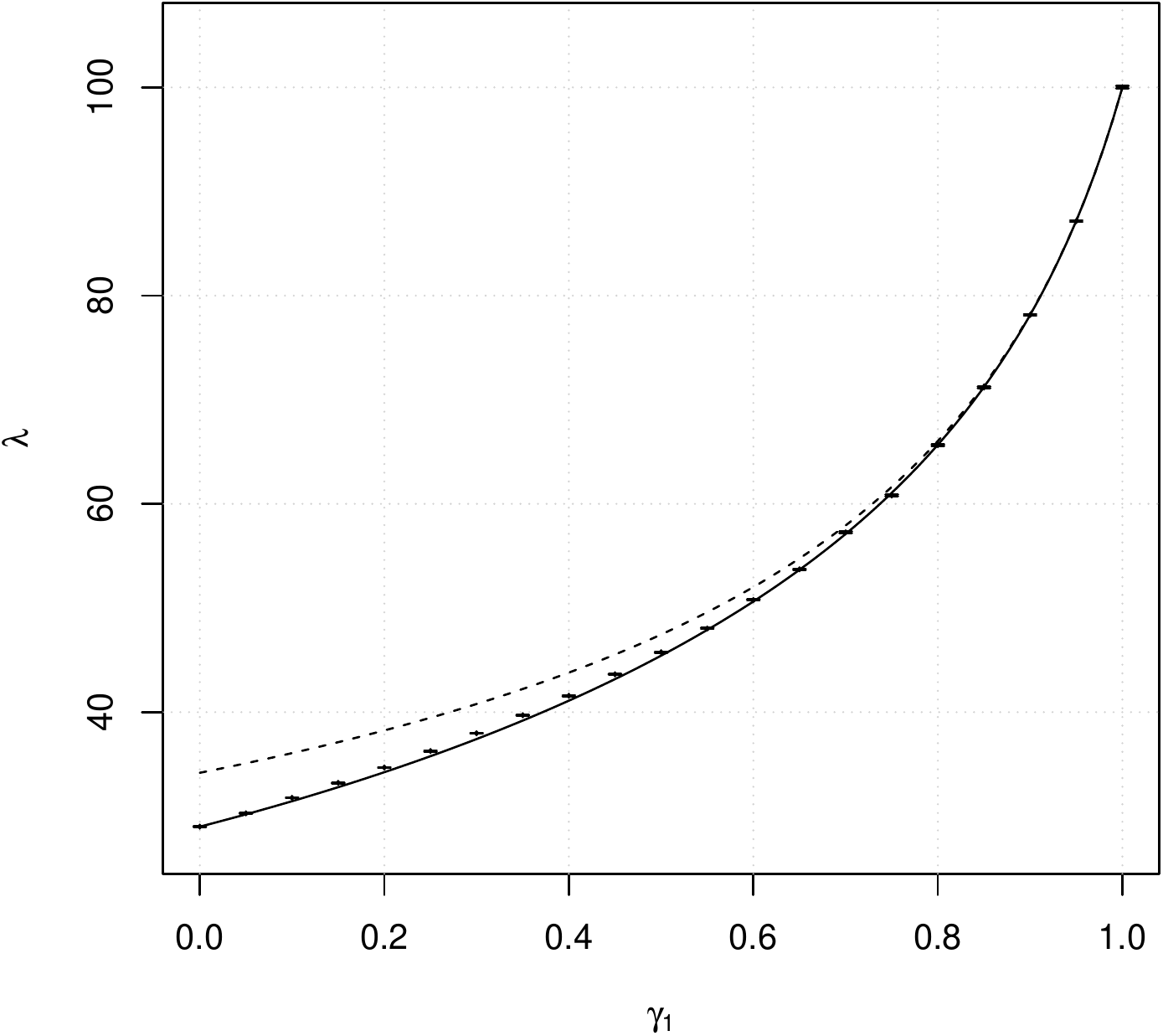}} 
\caption{Comparison of the exact intensity (small boxplots), the Poisson-saddlepoint approximation (dashed line) and the DPP approximation  (solid line) for homogeneous Strauss models with activity parameter $\beta$ and range of interaction $R$. Curves and boxplots are reported in terms of the interaction parameter $\gamma_1 \in [0,1]$. 
\label{fig:straussintro} }
\end{figure}

The rest of the paper is organized as follows. Section~\ref{sec:background} provides necessary notation and background material on point processes, Gibbs point processes and determinantal point processes. Intensity approximations are discussed in detail in Section~\ref{sec:approximation}. Finally, Section~\ref{sec:simu} presents numerical experiments for several classes of pairwise interaction point processes.

\section{Gibbs point processes and determinantal point processes} \label{sec:background}

\subsection{Background and Poisson point processes} \label{sec:poisson}

For $d\geq1$, let $\bX$ be a spatial point process defined on $\R^d$, which we see as a random locally finite subset of $\R^d$. Local finiteness of $\mathbf X$ means that 
$\mathbf X_B= \bX \cap B$ is finite almost surely (a.s.), that is the number of points $N(B)$ of $\bX_B$ is finite a.s., whenever
$B \subset \R^d$ is bounded. We let ${\cal N}$ stand for the state space consisting of 
the locally finite subsets (or point configurations) of $\R^d$. Let $\mathcal B(\R^d)$ denote the class of bounded Borel sets in~$\R^d$. For any $B\in \mathcal B(\R^d)$, we denote by $|B|$ its Lebesgue measure. A realization of $\bX_B$ is of the form
$\bx=\{x_1,\dots,x_m\}\subset B$ for some nonnegative finite integer $m$ and we sometimes denotes its cardinal by $|\bx|$. For further details about point processes, we refer to~\cite{daley:vere-jones:03} and \cite{moeller:waagepetersen:04}.  

A spatial point process is said to have an $n$th order intensity function $\rho^{(n)}$ if for any nonnegative measurable function $h: (\R^d)^n \to \R^+ $, the following formula referred to as Campbell-Mecke formula holds
\begin{align}
  \E \!\sum_{u_1,\dots,u_n \in \bX}^{\neq} h(u_1,\dots,u_n) 
= &\!\!\! \int_{\R^d} \!\dots \!\int_{\R^d}  h(u_1,\dots,u_n ) \rho^{(n)}(u_1,\dots,u_n) \dd u_1 \dots \dd u_n ,\label{eq:campbell}
\end{align}
where the sign $\neq$ over the sum means that $u_1,\dots,u_n$ are pairwise distinct. Then, \linebreak
$\rho^{(n)}(u_1,\ldots,u_n)\,\mathrm du_1\cdots\,\mathrm du_n$ can be
interpreted as the approximate probability for $\bX$ having a point in
each of infinitesimally small regions around $u_1,\ldots,u_n$ of
volumes $\mathrm du_1,\ldots\,\mathrm du_n$, respectively. We also write $\rho(u)$ for the intensity function $\rho^{(1)}(u)$. A spatial point process $\bX$ in $\R^d$ is said to be stationary (respectively isotropic) if its distribution is invariant under translations (respectively under rotations). When $\bX$ is stationary, the intensity function reduces to a constant denoted by $\lambda$ in the rest of this paper. As a matter of fact, $\lambda$ measures the mean number of points per unit volume. 

The Poisson point process, often defined as follows (see e.g. \cite{moeller:waagepetersen:04}), serves as the reference model.
\begin{definition}\label{Def Poisson}
Let $\rho$ be a locally integrable function on $S$, for $S\subseteq \R^d$. A point process $\bX$ satisfying the following statements is called the Poisson point process on $S$ with intensity function $\rho$:
\begin{itemize}
  \item for any $m\geq 1$, and for any disjoint and bounded $B_1,\dots,B_m  \subset S$, the random variables $\bX_{B_1},\dots,\bX_{B_m}$ are independent;
  \item $N(B)$ follows a Poisson distribution with parameter $\int_B \rho(u) \dd u$ for any bounded $B \subset S$.
\end{itemize}
\end{definition}
Among the many properties of Poisson point processes, it is to be noticed that the $n$th order intensity function writes $\rho^{(n)}(u_1,\dots,u_n) = \prod_{i=1}^n \rho(u_i)$, for any pairwise distinct $u_1,\dots,u_n \in S$.

Let $\bZ$ be a unit rate Poisson point process on $S$, which means that its intensity is constant and equal to one. Assume, first, that $S$ is bounded ($|S|<\infty$). We say that a spatial point process $\bX$ has a density $f$ if the distribution of $\bX$ is absolutely continuous with respect to the one of $\bZ$ and with density $f$. Thus, for any nonnegative measurable function $h$ defined on $\mathcal N$, $\E h(\bX) = \E(f(\bZ)h(\bZ))$.  Now, suppose that $f$ is {\it hereditary}, i.e., for any pairwise
distinct $u_0,u_1,\ldots,u_n\in S$, 
$f(\{u_1,\ldots,u_n\})>0$ whenever $f(\{u_0,u_1,\ldots,u_n\})>0$. 
We can then define the so-called {\it Papangelou conditional intensity} by 
\begin{equation}\label{eq:papangelou}
\lambda(u,\bx)={f(\bx \cup u)}/{f(\bx)}
\end{equation}
for any $u\in S$ and
$\bx \in\mathcal N$, setting
$0/0=0$. By the interpretation of $f$, $\lambda(u,\bx)\dd u$ can
be considered as the conditional probability of observing
one event in  a small ball, say $B$, centered at $u$ with volume $\dd u$, given that $\bX$ outside $B$ agrees with $\bx$. When $f$ is hereditary, there is a one-to-one correspondence between $f$ and $\lambda$.

Because the notion of density for $\bZ$  when $S=\R^d$ makes no sense, the Papangelou conditional intensity cannot be defined through a ratio of densities in $\R^d$. But it still makes sense as the Papangelou conditional intensity can actually be defined at the Radon-Nykodym derivative of $\PP_{u}^!$ the reduced Palm distribution of $\bX$ with respect to $\PP$, the distribution of $\bX$ (see \citet{daley:vere-jones:03}). We do not want to enter in too much detail here and prefer to refer the interested reader to \cite{coeurjolly:moeller:waagepetersen:17b}. 

Finally, we mention the celebrated Georgii-Nguyen-Zessin formula 
\citep[see][]{georgii:76,nguyen:zessin:79}, which states that for any $h:\R^d \times \mathcal N\to \R$ (such that the following expectations are finite)
 \begin{equation}
    \E \sum_{u \in \bX} h(u, \bX\setminus u) =  \int_{\R^d}\E \left(h(u,\bX) \lambda(u,\bX) \right)\dd u.\label{eq:GNZ}
 \end{equation}
By identification of~\eqref{eq:campbell} and~\eqref{eq:GNZ}, we see a link between the intensity function of a point process and the  Papangelou conditional intensity: for any $u\in \R^d$
\[
  \rho (u) = \E \left( \lambda(u, \bX)\right),
\]
which in the stationary case  reduces to 
\begin{equation} \label{eq:intensityGeneral}
  \lambda = \E \left( \lambda(0,\bX) \right).  
\end{equation}

\subsection{Gibbs point processes} \label{sec:gibbs}

For a recent and detailed presentation, we refer to \cite{dereudre:17}. Gibbs processes are characterized by an energy function $H$ (or Hamiltonian) that maps any finite point configuration to $\mathbb R \cup \{\infty\}$. Specifically, if $|S|<\infty$, a Gibbs point process on $S$ associated to $H$ and with activity $\beta>0$ admits the following density with respect to the unit rate Poisson process:
\begin{equation}\label{density Gibbs}
 f(\bx)  \propto  \beta^{|\bx|} e^{-H(\bx)},
\end{equation} 
where $\propto$ means ``proportional to''. This definition makes sense under some regularity conditions on $H$, typically non degeneracy ($H(\emptyset)<\infty$) and stability (there exists $A\in\mathbb R$ such that $H(\bx)\geq A |\bx|$ for any $\bx \in \mathcal N$). Consequently, configurations $\bx$ having a small energy $H(\bx)$ are more likely to be generated by a  Gibbs point process than by a Poisson point process, and conversely for configurations having  a high energy. In the extreme case where $H(\bx)=\infty$, then $\bx$ cannot, almost surely, be the realization of a Gibbs point process associated to $H$.

In this paper, we focus on pairwise interaction point processes. To be close to the original paper by \cite{baddeley:nair:12} the present contribution is based on, we use their notation: a Gibbs point process in $S$ is said to be a pairwise interaction point process with pairwise interaction function $g:\R^d \to \R^+$, if its density writes
\[
   f(\bx) \propto \beta^{|\bx|} \prod_{u,v \in \bx} g(u-v).
 \] 
If $|S|=\infty$, this definition and more generally Definition~\eqref{density Gibbs} do not make sense since $H(\bx)$ can be infinite or even undefined if $|\bx|=\infty$. In this case, Gibbs point processes have to be defined via their conditional specifications and for pairwise interactions Gibbs point processes, restrictions on $g$ have to be imposed for existence (see again \cite{dereudre:17} and the references therein for details). Nonetheless, as mentioned in the previous section, the concept of Papangelou conditional intensity applies whenever $|S|<\infty$ or $|S|=\infty$, and in either case it has the explicit form 
\begin{equation} \label{eq:papangelouPIPP}
  \lambda(u,\bx) = \beta \prod_{v\in \bx}  g(u-v),  
\end{equation}
for any $u \in S$. Note that when $S=\R^d$, a pairwise interaction Gibbs point process is stationary  if $g$ is symmetric and it is further  isotropic if $g(u-v)$ depends simply on $\|u-v\|$.

From~\eqref{eq:intensityGeneral}, we deduce that the intensity parameter of a stationary pairwise interaction process writes
\begin{equation}\label{eq:intensityGibbs}
  \lambda = \E \left( \lambda(0, \bX) \right) = \beta \E \left( \prod_{v\in \bx}  g(v) \right).
\end{equation}

Let us give a few examples (which are in particular well-defined in $\R^d$). Many other examples can be found e.g. in the recent monograph by \cite{baddeley:rubak:turner:15}.

\begin{itemize}
  \item {\it Strauss model}: let $\gamma \in [0,1]$ and $0<R<\infty$
  \begin{equation}\label{g_Strauss}
    g(u) = \left\{ \begin{array}{ll}
    \gamma & \mbox{ if } \|u\| \leq R  \\
    1 & \mbox{ otherwise.}
    \end{array} \right.
    \end{equation}
  \item {\it Strauss Hard-core model}:  let $\gamma \in \R^+$ and $0<\delta <R<\infty$
  \[
    g(u) = \left\{ \begin{array}{ll}
    0 & \mbox{ if } \|u\|< \delta\\
    \gamma & \mbox{ if } \delta\leq \|u\| \leq R  \\
    1 & \mbox{ otherwise.}
    \end{array} \right.
  \]
  \item {\it Piecewise Strauss Hard-core model}: 
  \[
    g(u) = \left\{ \begin{array}{ll}
    0 & \mbox{ if } \|u\|< \delta \\
    \gamma_i & \mbox{ if } R_i\leq \|u\| \leq R_{i+1},\, i=1,\dots,I  \\
    1 & \mbox{ otherwise,}
    \end{array} \right.
  \]
  where $I\geq 1$, $0\leq R_1=\delta<R_2<\dots<R_{I+1}=R<\infty$ and $\gamma_1,\dots,\gamma_I \in \R^+$ if $\delta>0$, otherwise $\gamma_1,\dots,\gamma_I \in [0,1]$.
  \item {\it Diggle-Graton model}: let $\gamma \in [0,1]$
  \[
    g(u) = \left\{ \begin{array}{ll}
      \left(\frac{\|u\|}{R}\right)^{1/\gamma} &\mbox{ if } \|u\| \leq R  \\
    1 & \mbox{ otherwise,}
    \end{array} \right.
  \]
  where for $t\in (0,1)$, $t^{\infty}=0$ and $1^\infty=1$ by convention. 
\end{itemize}

Let us note that a Strauss model with $\gamma=0$ and radius $R$ is actually a hard-core model with radius $R$. The Diggle-Graton potential can be found in \cite{baddeley:rubak:turner:15} in a slightly different parameterization. The one chosen here makes comparisons with the Strauss model easier. For instance, when $\gamma=0$ the model reduces to a Strauss model with $\gamma=0$ and radius $R$. When $\gamma=1$, the function $g$ grows linearly from 0 to 1. Figure~\ref{fig:models} depicts the form of some of the pairwise interaction functions presented above.

\begin{figure}[h]
\subfigure[Strauss model]{\includegraphics[scale=.28]{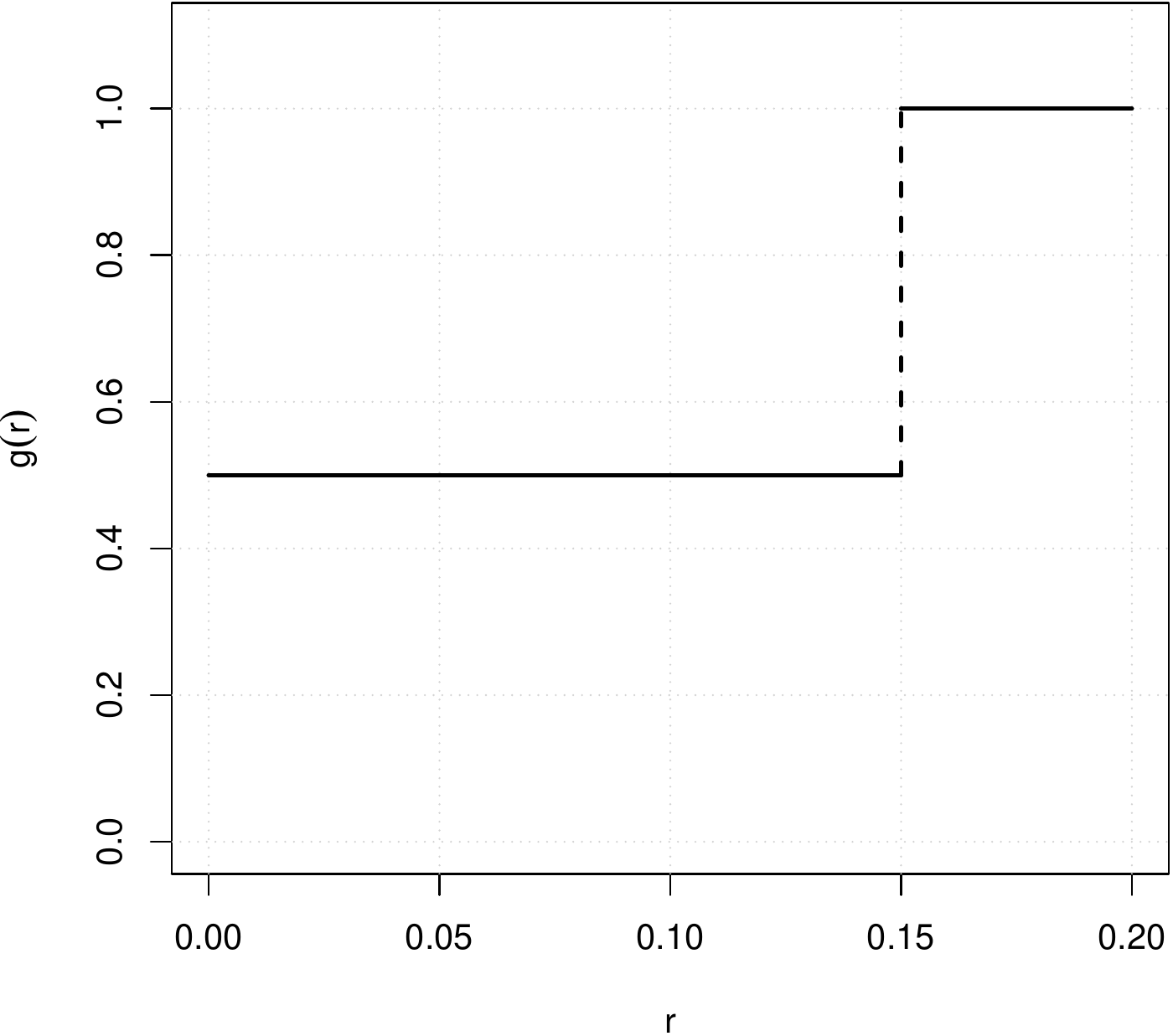}}
\subfigure[Piecewise Strauss hard-core model]{\includegraphics[scale=.28]{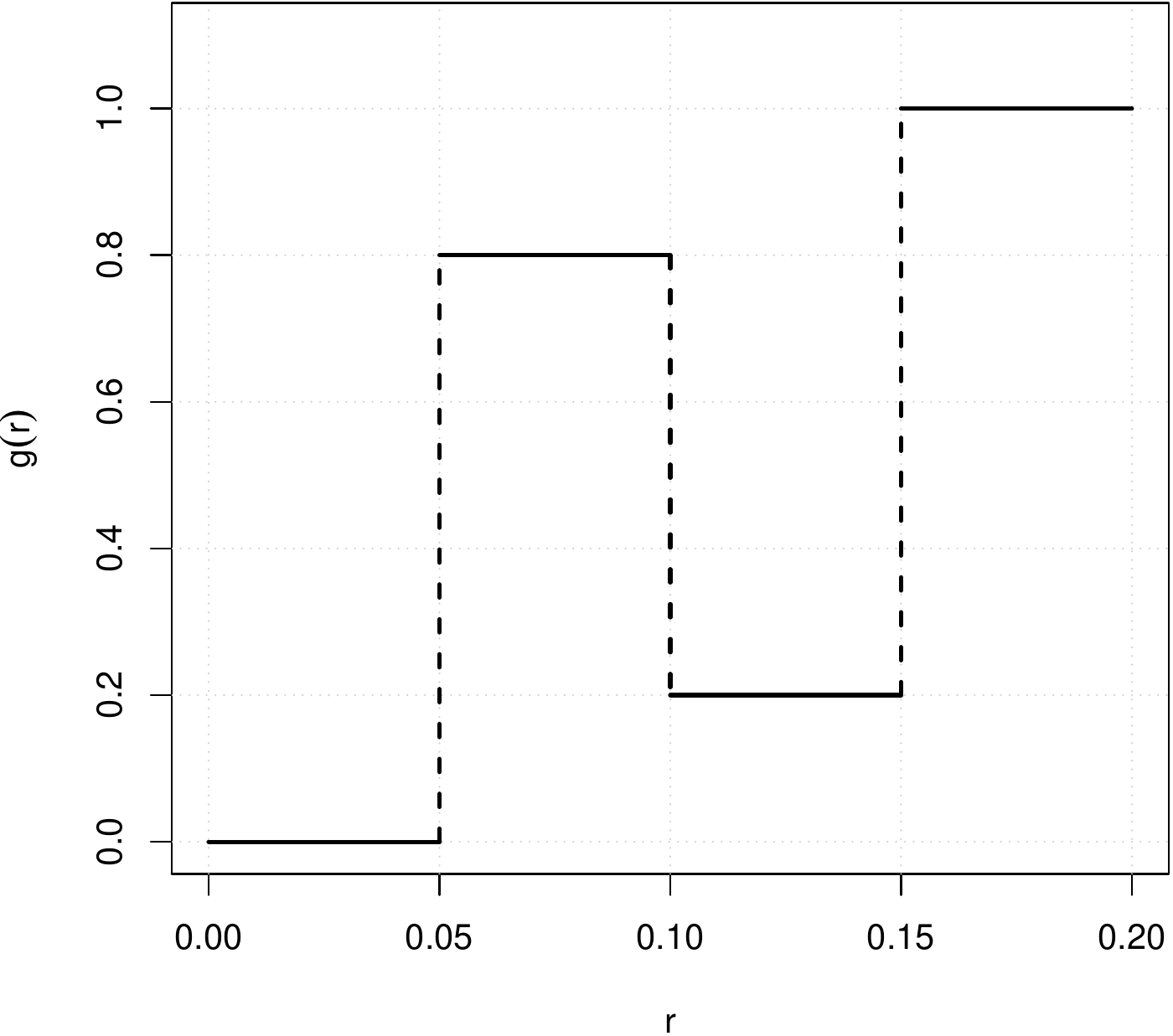}}
\subfigure[Diggle-Graton model]{\includegraphics[scale=.28]{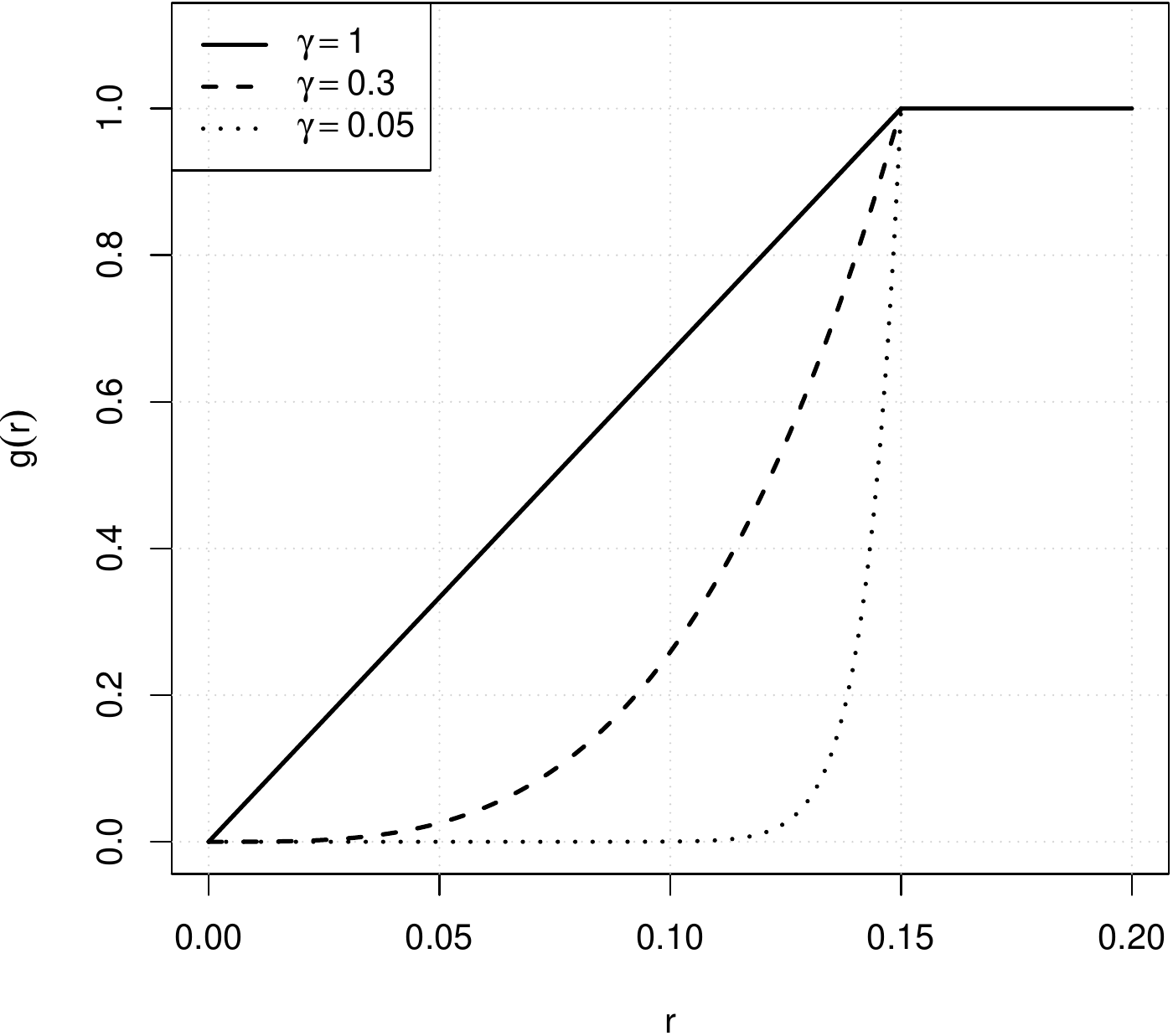}}
\caption{\label{fig:models} Examples of pairwise interaction functions for the Strauss model ($\gamma=0.5$, $R=0.15$), the piecewise Strauss hard-core model ($\delta=R_1=0.05$, $R_2=0.1$, $R_3=R=0.15$, $\gamma_1=0.8$, $\gamma_2=0.2$),  and the Diggle-Graton model ($\gamma=0.05,\,0.3$ and $1$, $R=0.15$).}
\end{figure}

A Gibbs point process has a finite range $R$ if for any $u\in \R^d$ and $\bx \in {\cal N}$, $\lambda(u,\bx) = \lambda(u,\bx \cap B(u,R))$. For pairwise interaction point processes, this property translates to $g(u) =1$ for any $u \in \R^d$ such that $\|u\|>R$. All previous models have a finite range $R<\infty$. An example of infinite range pairwise interaction point process which will not be considered in this paper is  the Lennard-Jones model (see e.g. \cite{ruelle:69,baddeley:rubak:turner:15}).

\subsection{Determinantal point processes} \label{sec:dpp}

Determinantal point processes (DPPs) are models for inhibitive point patterns. We refer to \cite{lavancier:moeller:rubak:15} for  their main statistical properties. They are defined through a kernel function $K$ which is a function from $S\times S$ to $\mathbb C$. A point process is a DPP on $S$ with kernel $K$, denoted by $\DPP(K)$, if for any $n$,  its $n$th order intensity function takes the form
\begin{equation}\label{def DPP}
  \rho^{(n)}(u_1,\dots,u_n) = \det[K](u_1,\ldots,u_n),
\end{equation}
for every $(u_1,\dots,u_n)\in S^n$, where $[K](u_1,\dots,u_n)$ denotes the matrix with entries $K(u_i,u_j)$,  $1\leq i,j\leq n$.
In particular, the intensity function of $\DPP(K)$ is $K(u,u)$.

Conditions on the kernel $K$ are required  to ensure the existence of  $\DPP(K)$. For our purpose, we will only consider DPPs on a compact set.  So let us assume that $S$ is compact and suppose that $K$ is a continuous real-valued covariance function on  $S\times S$. In this setting,  by the Mercer theorem (see \cite{riesz:90}), $K$ admits the spectral expansion 
\begin{align}\label{eq:spectral}
 K(u,v) = \sum_{i=1}^{\infty} \lambda_i \phi_i(u) {\phi_i(v)},\quad \forall u,v\in S,
\end{align}
where $\lbrace \phi_i \rbrace_{i\geq 1}$ is an orthonormal basis of $L^2(S)$ and where $\lambda_i$, 
$i\geq 1$, are referred to as the eigenvalues of $K$. Under the above assumptions, $\DPP(K)$ exists if and only if $\lambda_i\leq 1$ for all $i$. 

Due to their tractability, DPPs have many interesting properties. Many of them have been obtained by \cite{shirai:takahashi:03}, from which we derive the following key-equation used by our intensity approximation.

\begin{proposition} \label{prop:laplace}
Let $\bX$ be a $\DPP$ on a compact set $S$ with  kernel $K$.  Assume that $K$ is a continuous real-valued covariance function on $S\times S$ whose all eigenvalues are not greater than 1. 
For any function $g: S \to [0,1]$ 
\begin{equation}
  \label{eq:laplaceDPP}
  \E \left( \prod_{v\in \bX} g(v)\right) = \prod_{i\geq 1} (1-\tilde \lambda_i)
  \end{equation}
where $\tilde\lambda_i$, for $i\geq 1$, are the eigenvalues of the kernel $\widetilde K: S\times S \to \R$ given by
\[\tilde K(u,v)  = \sqrt{1 -g(u)} K(u,v)  \sqrt{1 - g(v)}.\]
\end{proposition}
\begin{proof}
Note that
\[ \E \left( \prod_{v\in \bX} g(v)\right) = L_{\bX}(-\log g)\]
where $L_{\bX}$ denotes the Laplace transform of $\bX$. 
From Theorem 1.2 in  \citet{shirai:takahashi:03}, for any nonnegative measurable function $f$ on $S$ 
\[ L_{\bX}(f) = {\rm Det}  ( I -\mathcal {\tilde K})\]
where ${\rm Det}$ denotes the Fredholm determinant of an operator and $\mathcal {\tilde K}$ is the integral operator associated to the kernel
\[\tilde K(u,v)  = \sqrt{1 -\exp(-f(u))} K(u,v)  \sqrt{1 - \exp(-f(v))}.\]
On the other hand, see for instance (2.10)  in  \citet{shirai:takahashi:03},
\[ {\rm Det}  ( I -\mathcal {\tilde K}) = \exp\left(-\sum_{n=1}^\infty \frac 1n {\rm Tr}\, (\mathcal {\tilde K}^n )\right),\]
where ${\rm Tr}$ denotes the trace operator.  The result follows from the fact that for any $n\geq 1$
\[
  {\rm Tr}\, (\mathcal {\tilde K}^n) =\int_{S^n} \tilde K(u_1,u_2)\cdots \tilde K(u_n,u_1) {\rm d} u_1\cdots {\rm d} u_n = \sum_{i\geq 1} \tilde \lambda_i^n.
\]
\end{proof}

\section{Intensity approximation} \label{sec:approximation}

\subsection{Poisson-saddlepoint approximation}  \label{sec:lps}

We remind that the intensity parameter of a Gibbs point process, and in particular a pairwise interaction point process satisfies~\eqref{eq:intensityGibbs}. The expectaction in~\eqref{eq:intensityGibbs} is to be regarded with respect  to $\mathrm P$ the distribution of the Gibbs point process~$\bX$. \citet{baddeley:nair:12} suggest to replace $\mathrm P$ by a simpler distribution, say $\mathrm Q$, for which the right-hand-side of~\eqref{eq:intensityGibbs} becomes tractable. The Poisson-saddlepoint approximation consists in choosing $\Pi(\lambda)$, the Poisson distribution with parameter $\lambda$, as distribution $\mathrm Q$.
As a result, the Poisson-saddlepoint approximation consists in resolving the equation
\begin{equation} \label{eq:PSapproximation}
  \lambda = \beta \, \E_{\Pi(\lambda)} \left(  \prod_{v \in \bY} g(v)\right) = \beta \E_{\Pi(\lambda)} \left( \exp\left( \sum_{v \in \bY} \log g(v) \right) 
  \right),
\end{equation}
with the convention that $\log 0=-\infty$ and where, to avoid any ambiguity, we denote by $\bY$ a Poisson point process with intensity $\lambda$ defined on $\R^d$ and stress also this by indexing the $\E$ with the distribution $\Pi(\lambda)$. 
It turns out that if $g(u)\in[0,1]$ for any $u\in \R^d$, the right-hand side of~\eqref{eq:PSapproximation} is the Laplace transform of some Poisson functional and equals $\beta\exp(-\lambda G)$ where $G = \int_{\R^d} (1-g(u))\dd u$, see e.g. \citet[Proposition 3.3]{moeller:waagepetersen:04}. As noticed in \citet{baddeley:nair:12}, this formula extends to more general functions $g$, provided $G>-\infty$.  Hence, the Poisson-saddlepoint approximation, denoted by $\lambda_\PS$ in this paper, is defined as the solution of
\begin{equation}
  \label{eq:lPS}
  \lPS = \beta \exp(-\lPS  \, G)  \quad \Longleftrightarrow \quad
  \lPS = \frac{W(\beta G)}{G}
\end{equation}
where $W$ is the inverse function of $x \mapsto x \exp(x)$.

For stationary pairwise Gibbs models with finite range $R$, and such that $\lambda(u,\bx)\leq \beta$ (or equivalently such that $g\leq 1$), then $0\leq G\leq |B(0,R)|$. In this case, \citet{baddeley:nair:12} prove, among other properties, that  $\lPS$  exists uniquely and is an increasing function of $\beta$.
From a numerical point of view, $\lambda_\PS$ can be very efficiently and quickly estimated using root-finding algorithms.

\subsection{DPP approximation} \label{sec:ldpp}

Following the same idea as the Poisson-saddlepoint approximation, for a repulsive stationary pairwise interaction point process  with pairwise interaction function $g\leq 1$ having a finite range $R$, we suggest to substitute the measure $\mathrm P$ involved in the expectation~\eqref{eq:intensityGibbs} by the measure $\mathrm Q$ corresponding to a $\DPP$ defined on $B(0,R)$ with some kernel $K$ (to be chosen) and intensity $\lambda$, i.e. $K(u,u)=\lambda$. Similarly to the previous section, by letting $\DPP(K;\lambda)$ denote the distribution of such a $\DPP$ and $\bY \sim \DPP(K;\lambda)$, the DPP approximation  of the intensity $\lambda$  is the solution of
\begin{equation}\label{eq1_DPP} \lambda = \beta \, \E_{\DPP(K;\lambda)}\left(  \prod_{v \in \bY} g(v)\right).
\end{equation}

From Proposition~\ref{prop:laplace} and in particular from~\eqref{eq:laplaceDPP},  this yields the estimating equation 
\[\log\lambda = \log\beta +\sum_{i\geq 1} \log (1-\tilde \lambda_i),\]
where the eigenvalues $\tilde \lambda_i$ of $\tilde K$ are related to $\lambda$ by the relation  
\[
\tilde K(u,v)=\sqrt{1 -g(u)} K(u,v)  \sqrt{1 - g(v)} \quad \mbox{ with } \quad K(u,u)=\lambda.   
\]
To complete this approximation, the eigenvalues $\tilde\lambda_i$ need to be specified.

In the following we  choose the eigenvalues $\tilde\lambda_i$ to be zero except a finite number $N$ of them that are all equal. Given that 
\[
\sum_{i\geq 1} \tilde\lambda_i= \int_{\R^d} \tilde K(u,u){\rm d} u = \int_{\R^d} (1-g(u))K(u,u){\rm d} u =\lambda G,  
\] 
this means that for some $N\geq \lambda G$,  
\begin{equation}\label{eq:tilde_lambda}
\tilde \lambda_i = \frac{\lambda G} N,\quad \text{for } i=1,\dots,N
\end{equation}
 and $\tilde \lambda_i=0$ for $i\geq N+1$. With this choice, the integer $N$ remains the single parameter to choose in our approximation. Note that $N\geq \lambda G$ is a necessary condition to ensure $\tilde\lambda_i \leq 1$ and so the existence of a DPP with kernel $\tilde K$, but it is  in general not sufficient to ensure the existence of the relation between $\tilde K$ and  $K$ where $K$ defines a DPP. This will be clearly illustrated below when $g$ is the Strauss interaction function.   For the choice \eqref{eq:tilde_lambda}, the DPP approximation of the intensity, denoted by $\lDPP$,  becomes the solution of 
\begin{equation}\label{eq:estimating}\log\lDPP = \log\beta + N \log \left(1- \frac{\lDPP G} N\right)
\Longleftrightarrow
\lDPP = \beta \left( 1 - \frac{\lDPP G}N\right)^N.
\end{equation}

\medskip

To motivate \eqref{eq:tilde_lambda} and how we should set $N$, assume for a moment that $g$ is the interaction function of a Strauss model with range $R$ and interaction parameter $\gamma\in[0,1]$, see~\eqref{g_Strauss}. In this case
$\tilde K(u,v)=(1-\gamma) K(u,v)$ for any $u,v\in B(0,R)$ and the eigenvalues $\lambda_i$ of $K$ satisfy $\tilde\lambda_i = (1-\gamma) \lambda_i$.  
In the approximation \eqref{eq1_DPP}, we start by choosing a kernel $K$ with a finite number of non-vanishing eigenvalues $\lambda_i$ that are all equal.  In view of $\sum \lambda_i = \int K(u,u){\rm d} u = \lambda b$, where $b$ denotes the volume of $B(0,R)$, this leads to  
$\lambda_i = \lambda b/N$ for $i=1,\dots,N$ and $N\geq \lambda b$.  Note that the latter inequality is necessary to ensure the existence of $\DPP(K)$.  Going back to $\tilde\lambda_i$, this means that \eqref{eq:tilde_lambda} follows with the necessary and sufficient condition $N\geq \lambda b = \lambda G/(1-\gamma)$ which is greater than $\lambda G$. 

In order to set $N$ precisely for the Strauss model, remember that a homogeneous DPP is  more repulsive when  its eigenvalues  are close to 1, see~ \cite{lavancier:moeller:rubak:15, biscio2016}, and at the opposite a DPP is close to a Poisson point process when its eigenvalues are all close to 0. This suggests that in order to make the approximation \eqref{eq1_DPP} efficient, we should choose $\lambda_i$ close to 1 when the Gibbs process we want to approximate is very repulsive, that is when $\gamma$ is  close to 0. Moreover the eigenvalues should decrease to $0$ when $\gamma$ increases to $1$. If $\lambda_i = \lambda b/N$, this is equivalent to choosing $N$ an integer that increases from $\lambda b$ to infinity  when $\gamma$ increases from $0$ to $1$. A natural option is thus to choose $N$ as the smallest integer larger than  $\lambda b/(1-\gamma)$. Our final choice for the Strauss model is therefore  $N=\lceil \lambda b/(1-\gamma) \rceil$, where $\lceil.\rceil$ denotes the ceiling function, which we may write, for later purposes, 
$N=\lceil \lambda G/(1-\gamma)^2 \rceil$. 

However, with the latter choice,  the function in the right-hand side of equation \eqref{eq:estimating} is not continuous in $\lambda$, which may lead to none or several solutions to this equation. As a last step in our approximation, we therefore consider the upper convex envelope of this function, ensuring a unique solution to \eqref{eq:estimating}. This finally leads for the Strauss interaction process to the approximation $\lDPP$ defined as the solution of
\[
\log\lDPP = \log\beta + (1+ \lDPP G/(1-\gamma)^2 )  \log  \left(1- \frac{\lDPP G} {1+\lDPP G/(1-\gamma)^2}\right).
\]

\medskip

Let us now discuss the case of a general pairwise interaction function $g$. In this setting, it is  in general not possible  to relate the eigenvalues $\lambda_i$ of  $K$ with the eigenvalues $\tilde \lambda_i$  of  $\tilde K$. Motivated by the Strauss case, we choose $\tilde\lambda_i$ as in  \eqref{eq:tilde_lambda} where $N=\lceil \lambda G/\kappa \rceil$ and $\kappa \in [0,1]$ is a parameter that takes into account the repulsiveness encoded in $g$. 
In general $\kappa$ must be close to $0$ when $g$ is close to $1$ (the Poisson case), and close to $1$ when $g$ is close to a pure hard-core interaction. We decide to quantify the repulsiveness of the model by $b^{-1}\int (1-g)^2$, in agreement with our choice for the Strauss model for which $\kappa=(1-\gamma)^2$. Note that for a pairwise interaction $g$ with range $R$ and involving  a possible hard-core distance $\delta$, we have $|B(0,\delta)|\leq \int (1-g)^2\leq |B(0,R)|$ where the left or right equality occurs for a pure hard-core interaction (if $\delta>0$), a situation where $\kappa$ must be $1$. This leads us to the choice
\begin{equation}\label{eq:kappa}
\kappa = \max \left(\frac{|B(0,\delta)|}{\int (1-g)^2} , \frac{\int (1-g)^2}{|B(0,R)|}\right).
\end{equation}
Plugging $N=\lceil \lambda G/\kappa \rceil$ into  \eqref{eq:estimating} and considering the upper convex envelope to ensure the existence of a unique solution, we finally end up with our general DPP approximation being the solution of 
\begin{align}
&\log\lDPP = \log\beta + (1+\lDPP G/\kappa )  \log  \left(1- \frac{\lDPP G} {1+\lambda G/\kappa}\right)\label{eq:DPPapproximation1} \\
\Longleftrightarrow \quad & \lDPP = \frac{W_\kappa \left( \beta G/\kappa\right)}{G/\kappa}, \label{eq:DPPapproximation3}
\end{align}
where $\kappa$ is given by \eqref{eq:kappa} and
where $W_\kappa$ is the inverse function of \linebreak$x \mapsto x\left( 1 - \frac{\kappa x}{1+x}\right)^{-1-x}$.

In view of~\eqref{eq:DPPapproximation1}-\eqref{eq:DPPapproximation3} and similarly to $\lPS$, the approximation $\lDPP$ can be very efficiently implemented using root-finding algorithms. We further have the following properties. 

\begin{theorem}\label{thm:DPP}
Consider a stationary pairwise interaction process in $\R^d$ with Papangelou conditional intensity given by~\eqref{eq:papangelouPIPP} which is purely inhibitory, i.e. $g(u)\leq 1$ for all $u\in \R^d$ and with finite range $R$. Then,  $\lDPP$ exists uniquely, is an increasing function of $\beta$ and is such that $\lDPP \leq \lPS$. \\
 \end{theorem}

\begin{proof}
Let $f_\sPS$ and $f_\sDPP$ denote the two real-valued functions given by
\[
  f_\sPS(\lambda) = \beta \exp(-\lambda G) \quad \mbox{ and } \quad 
  f_\sDPP(\lambda) = \beta \left(1 - \frac{\lambda G}{1+\lambda G/\kappa} \right)^{1+\lambda G/\kappa}.
\]
The approximations $\lPS$ and $\lDPP$ are defined by the fixed point equations $\lPS=f_\sPS(\lPS)$ and $\lDPP=f_\sDPP(\lDPP)$. Since for any $x\in [0,1)$, $\log(1-x)\leq -x$, we have for any $\lambda$
\begin{equation}\label{eq:fPSfDPP}
  0 \leq f_\sDPP(\lambda) \leq \beta \exp(-\lambda G) = f_\sPS(\lambda).
\end{equation}
In particular $0\leq \lim_{\lambda\to \infty}f_\sDPP(\lambda) \leq \lim_{\lambda \to \infty} f_\sPS(\lambda)=0$. In addition, $f_\sDPP(0)=\beta$ and it can be verified that $f_\sDPP$ is a decreasing function. Hence the solution to~\eqref{eq:DPPapproximation1} exists uniquely. The function $W_\kappa$ can also be shown to be increasing on $\R^+$ for any $\kappa \in [0,1]$, so we deduce from~\eqref{eq:DPPapproximation3} that $\lDPP$ is an increasing function of $\beta$. Finally,~\eqref{eq:fPSfDPP} shows that $\lDPP\leq \lPS$. \\
\end{proof}

\section{Numerical study}\label{sec:simu}

In this section, we focus on the planar case to investigate the performances of the DPP approximation and compare it with the initial one proposed by \citet{baddeley:nair:12}. All computations were performed in the \texttt R language \citep{rcore:11}. The Poisson-saddlepoint approximation as well as the DPP approximation are implemented using root-finding algorithms and in particular we use the \texttt R function \texttt{uniroot} for this task.

We have considered 14 different numerical experiments involving Strauss models (S), Strauss hard-core models (SHC), Diggle-Graton models (DG), piecewise Strauss models (PS) and piecewise Strauss hard-core (PSHC) models. The pairwise interaction functions of these models are detailed in Section~\ref{sec:gibbs}. To sum up here are the parameters, that include a continuously varying parameter $\gamma_1\in [0,1]$:
\begin{itemize}
\item Strauss (S): $\beta=100$ with $R=0.05$ or $0.1$; $\beta=50$ with $R=0.1$ or $0.15$; $\beta=200$ with $R=0.05$. For all these models $\gamma=\gamma_1$.
\item Strauss hard-core (SHC): $\beta=200$, $\delta=0.025$, $R=0.05$. For this model $\gamma=\gamma_1$.
\item Diggle-Graton (DG): $\beta=200$, $R=0.025, 0.05$ or $0.075$ and $\beta=50$ and $R=0.15$. For all these models $\gamma=\gamma_1$.
\item Piecewise Strauss and Strauss hard-core (PS and PSHC): $\beta=200$, $\delta=0$ or $0.025$, $\gamma=(\gamma_1,\gamma_2)$ with $\gamma_2=0$ or $0.5$. The vector of breaks is $R=(0.05,0.1)$.
\end{itemize}

For each numerical experiment, we therefore obtain curves of intensity approximation in terms of $\gamma_1$. For $\gamma_1$ varying from 0 to 1 by step of 0.05 (the value 0 is exluded for DG models to save time), the true intensity $\lambda$ is estimated by Monte-Carlo methods. For each set of parameters $m$ realizations of the model are generated on the square $[-2R,1+2R]^2$ and then clipped to the unit square. That strategy is detailed and justified by \citet{baddeley:nair:12}. Specifically, the number of points in each realization is
averaged to obtain the estimated intensity and its standard error. The simulation results for the Strauss models with $\beta=50$ or $100$ were obtained by~\citet{baddeley:nair:12}, where $m=10000$ realizations were generated and the exact simulation algorithm was used, implemented in the \texttt{R} function \texttt{rStrauss} of the \texttt{spatstat} package (see \cite{baddeley:rubak:turner:15}). For the Strauss models with $\beta=200$, SHC models, PS and PSHC models, we generate $m=1000$ replications and use the \texttt{rmh} function in the \texttt{spatstat} package which implements a Metropolis-Hastings algorithm. Even if we use $10^6$ iterations of the algorithm, the results may be slightly biased. For the DG models, the \texttt{R} package \texttt{spatstat} provides an exact simulation algorithm (function \texttt{rDiggleGraton}) and for such models we generate $10000$ replications when $\beta=200$ and $R=0.025,0.05$ and when $\beta=50$ and $R=0.15$. We used  $1000$ replications when $\beta=200$ and $R=0.075$ to save time.

All results can be found in Figures~\ref{fig:strauss},~\ref{fig:dg} and~\ref{fig:piecewise}. Plots provide the same information: we depict intensity approximation $\lambda$ based on different methods in terms of $\gamma_1$. 
The dashed curve represents the Poisson-saddlepoint approximation proposed by \citet{baddeley:nair:12} and detailed in Section~\ref{sec:lps}. The solid curve  is the DPP approximation we propose in this paper and is given by~\eqref{eq:DPPapproximation3}.

Let us first comment Figure~\ref{fig:strauss} dealing with Strauss models. As expected the Poisson-saddlepoint approximation is not efficient when $\gamma_1$ is small, i.e. for very repulsive models. This is very significant in particular for the Strauss hard-core model, see Figure~\ref{fig:strauss}~(f). The DPP aproximation we propose is more likely able to capture the repulsiveness of the Strauss models. Figure~\ref{fig:dg} also clearly shows that our approximation is particulalry efficient and outperforms unamibigously the Poisson-saddlepoint approximation. Note that replications for the Diggle-Graton models are generated using an exact algorithm; so the numerical results seem to be exact, except the slight bias induced by clipping the pattern from $[-2R,1+2R]^2$ to the unit square. 

We finally comment Figure~\ref{fig:piecewise}. When $\gamma_2=0.5$, i.e. Figures~\ref{fig:piecewise}~(a)-(b), the results are very satisfactory. Our approximation is able to approximate $\lambda$ very efficiently for any value of $\gamma_1$. For Figures~\ref{fig:piecewise}~(c)-(d), $\gamma_2=0$ which means that points within a distance comprised between 0.05 and 0.1 are forbidden. Such a parameterization tends to create repulsive clusters. When $\gamma_1=1$ and $\delta=0$, such a piecewise Strauss model was called annulus model by \citet{stucki:schuhmacher:14}. This model demonstrates the limitations of our approximation even if when $\gamma_1$ is close to zero which means that the model is close to a hard-core process with radius $0.1$ our approximation remains satisfactory.

\begin{figure}[h]
\subfigure[S: $R=0.05$, $\beta=100$]{\includegraphics[scale=.4]{strauss005b100-1.pdf}}
\subfigure[S: $R=0.1$, $\beta=100$]{\includegraphics[scale=.4]{strauss01b100-1.pdf}} \\
\subfigure[S: $R=0.1$, $\beta=50$]{\includegraphics[scale=.4]{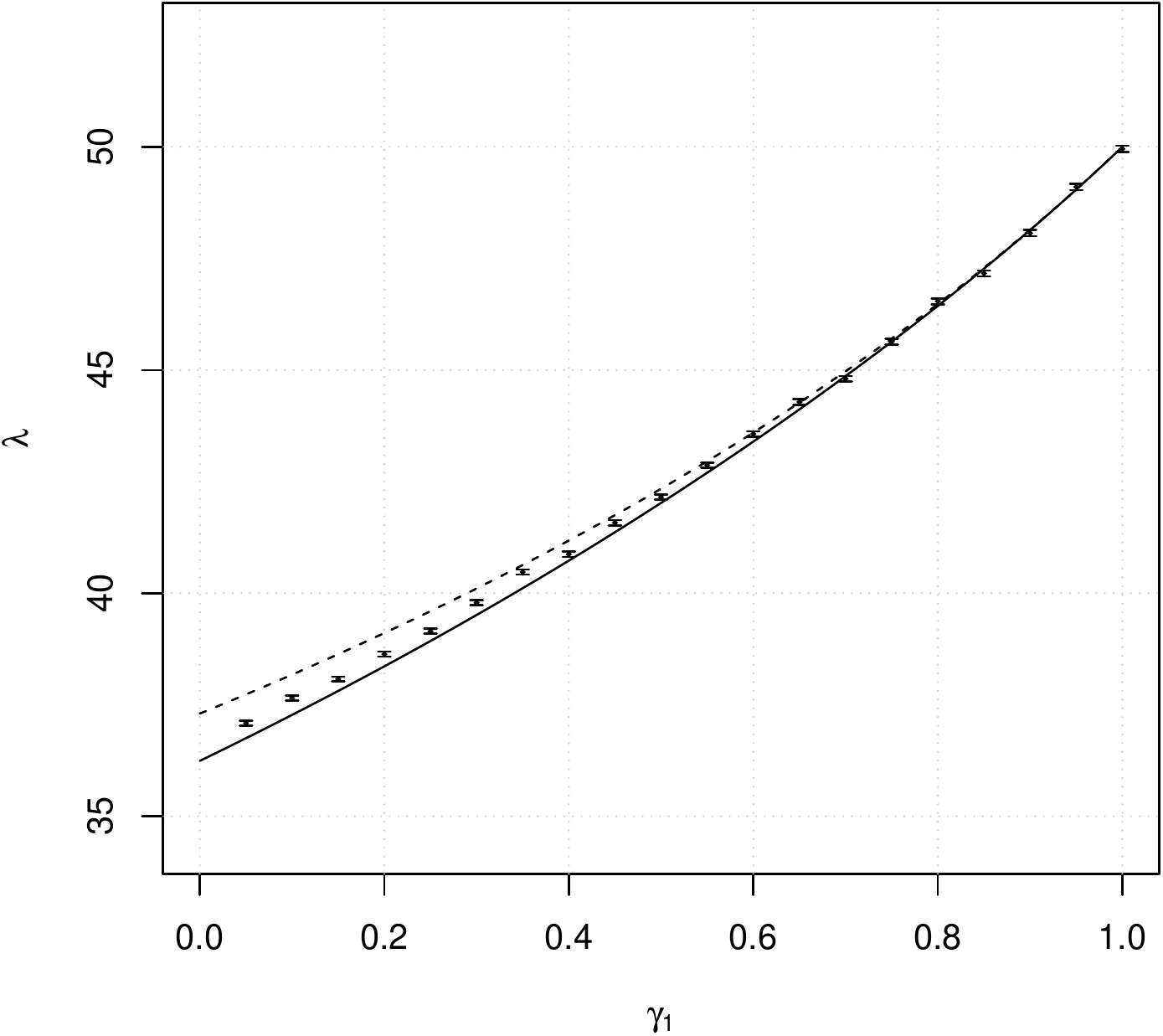}}
\subfigure[S:  $R=0.15$, $\beta=50$]{\includegraphics[scale=.4]{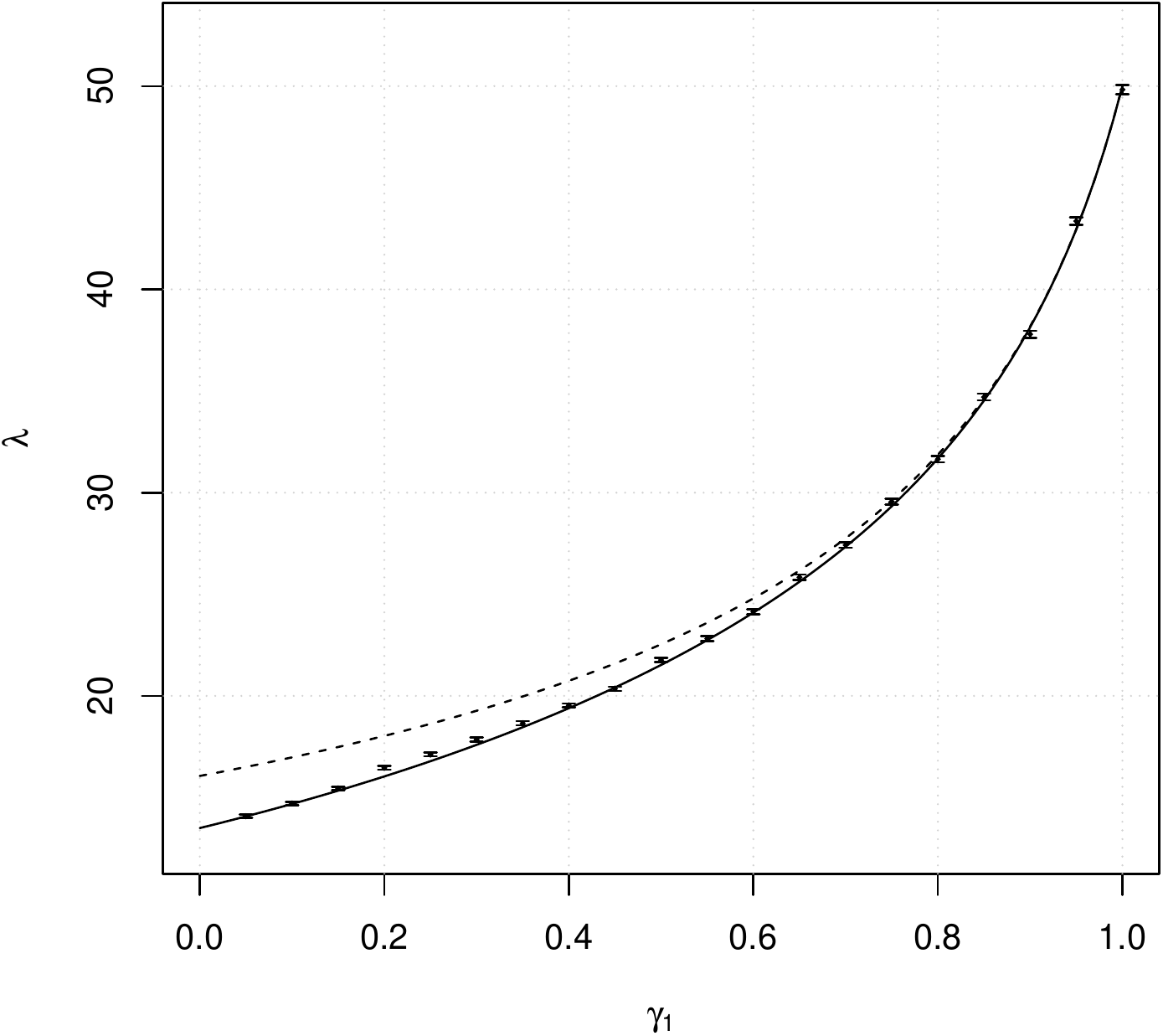}} \\
\subfigure[S: $R=0.05$, $\beta=200$]{\includegraphics[scale=.4]{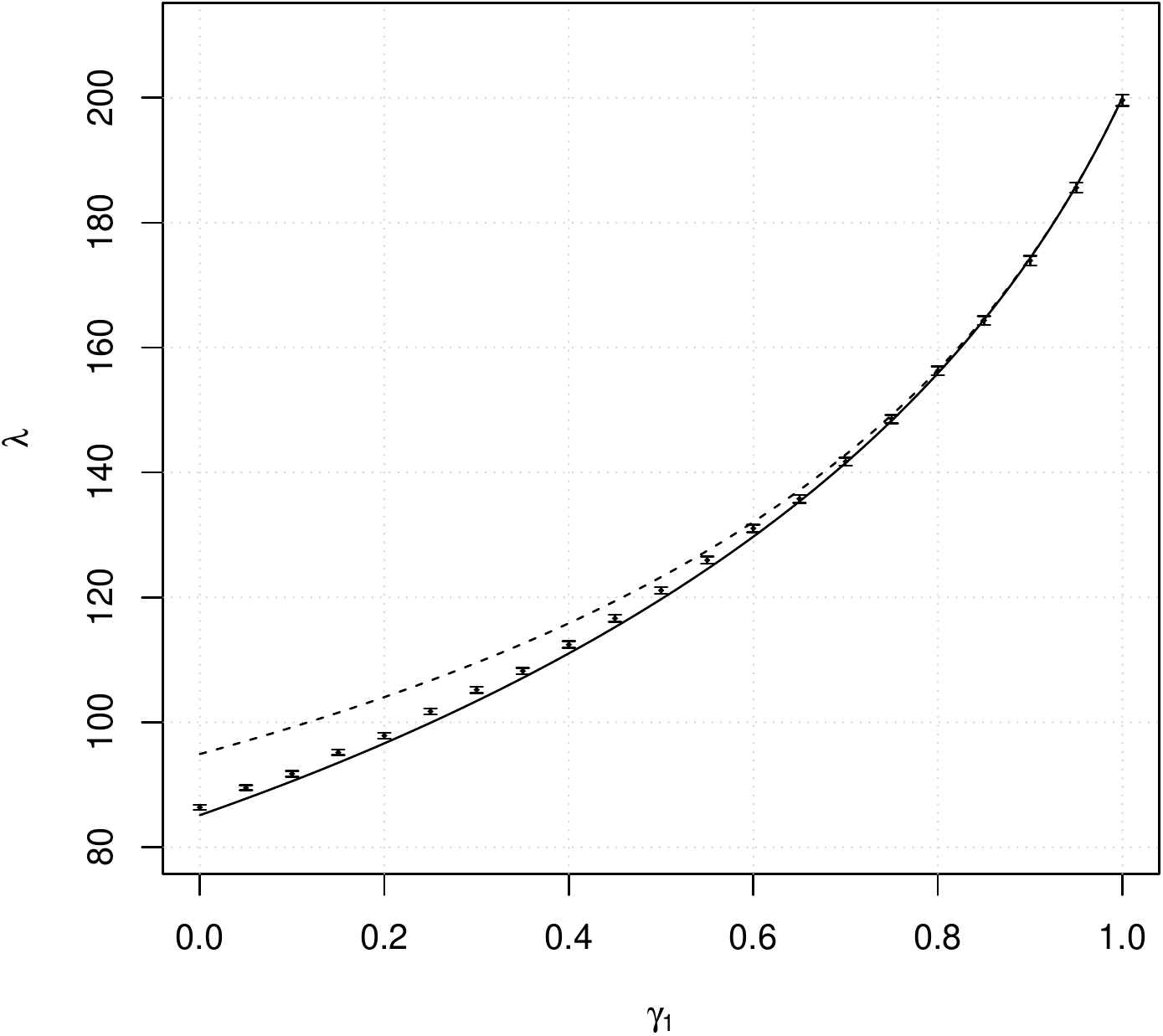}}
\subfigure[SHC: $\delta=0.025$, $R=0.05$, $\beta=200$]{\includegraphics[scale=.4]{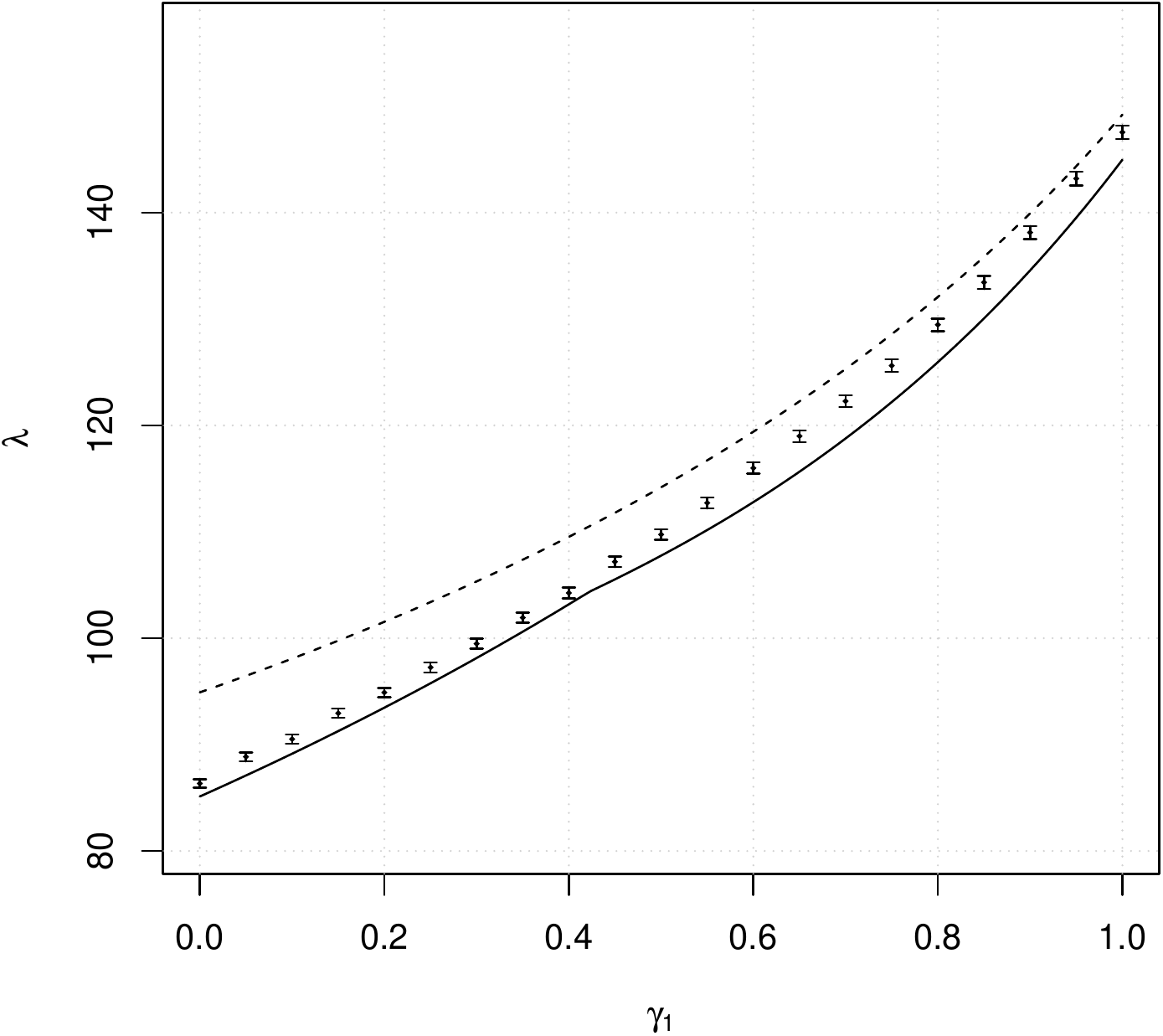}} \\
\caption{Comparison of the exact intensity (small boxplots obtained by Monte-Carlo method), the Poisson-saddlepoint approximation (dashed line) and the DPP approximation  (solid line) for homogeneous Strauss and Strauss hard-core models with activity parameter $\beta$, range of interaction $R$ and eventually hard-core distance $\delta$. Curves and boxplots are reported in terms of the interaction parameter $\gamma_1 \in [0,1]$. \label{fig:strauss} }
\end{figure}

\begin{figure}[h]
\subfigure[DG: $R=0.025$, $\beta=200$]{\includegraphics[scale=.4]{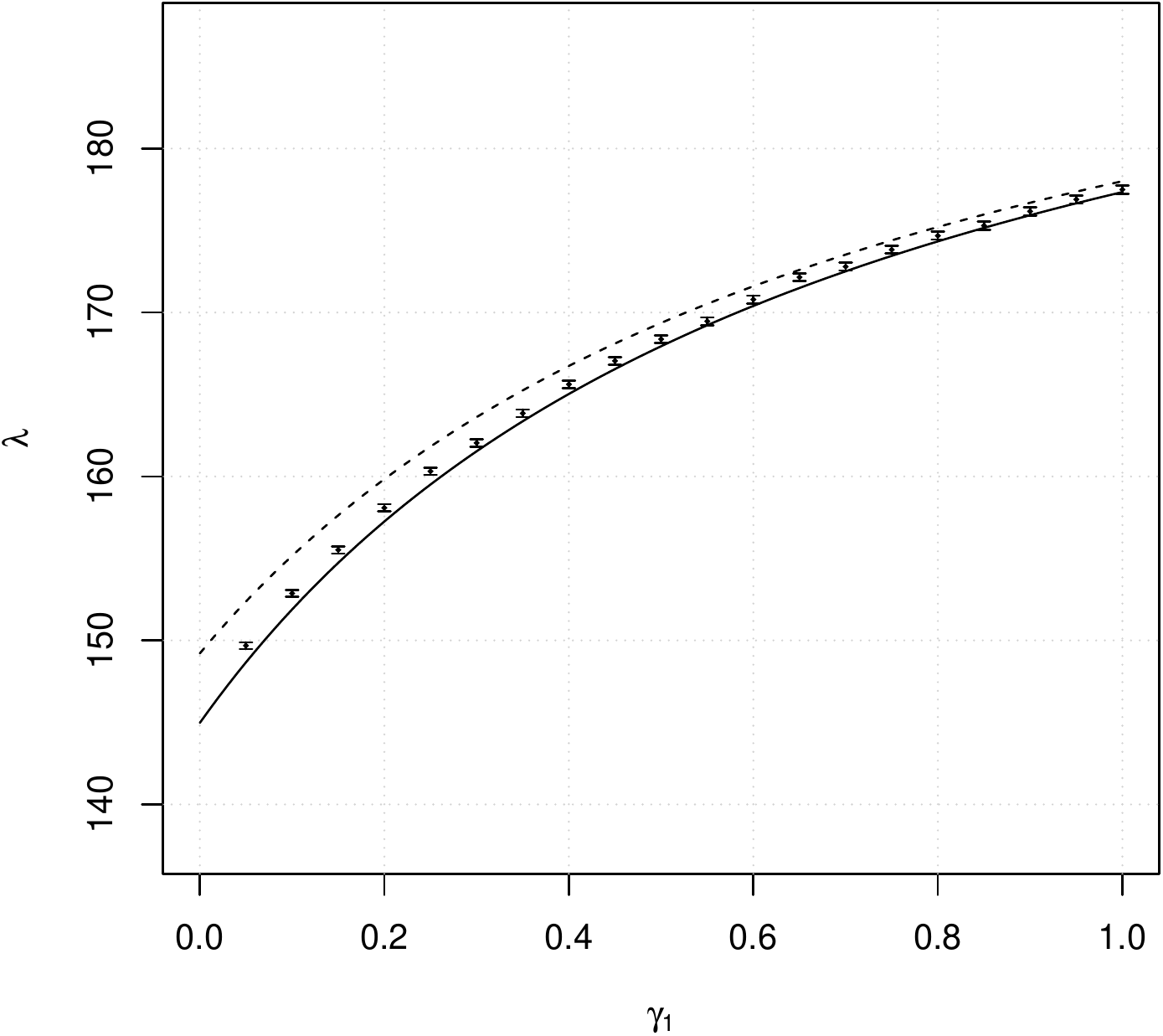}}
\subfigure[DG: $R=0.05$, $\beta=200$]{\includegraphics[scale=.4]{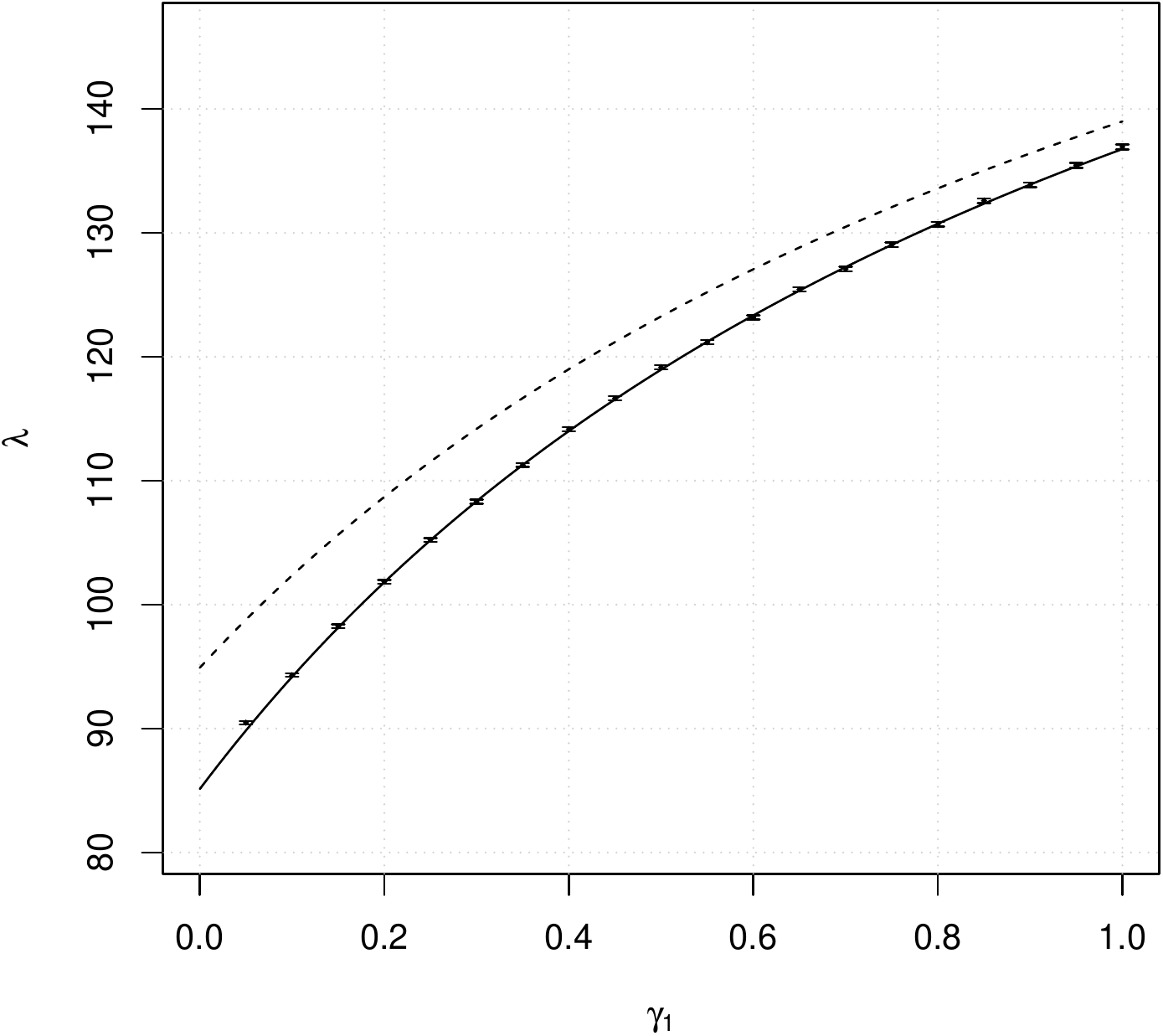}}\\
\subfigure[DG: $R=0.075$, $\beta=200$]{\includegraphics[scale=.4]{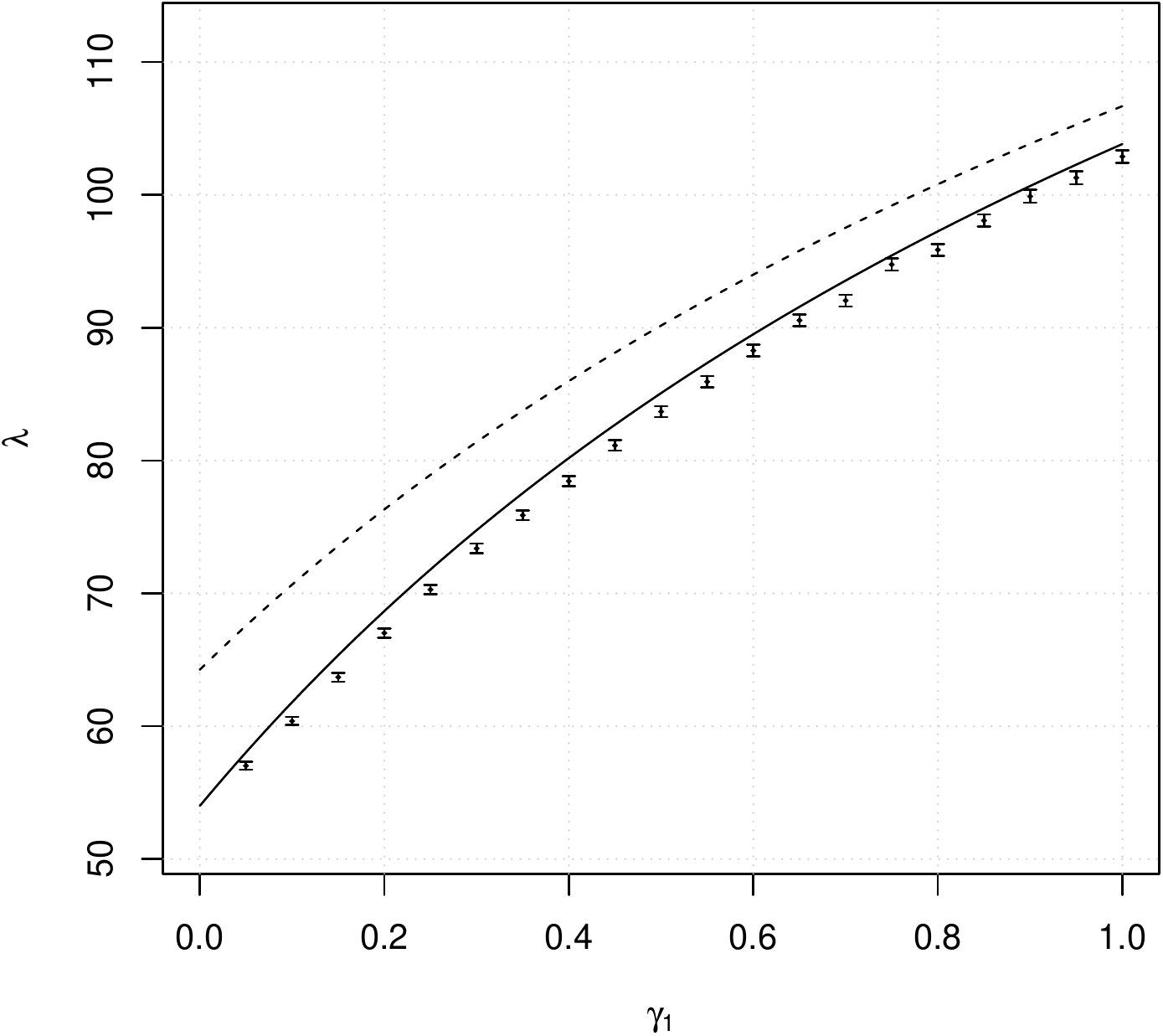}}
\subfigure[DG: $R=0.15$, $\beta=50$]{\includegraphics[scale=.4]{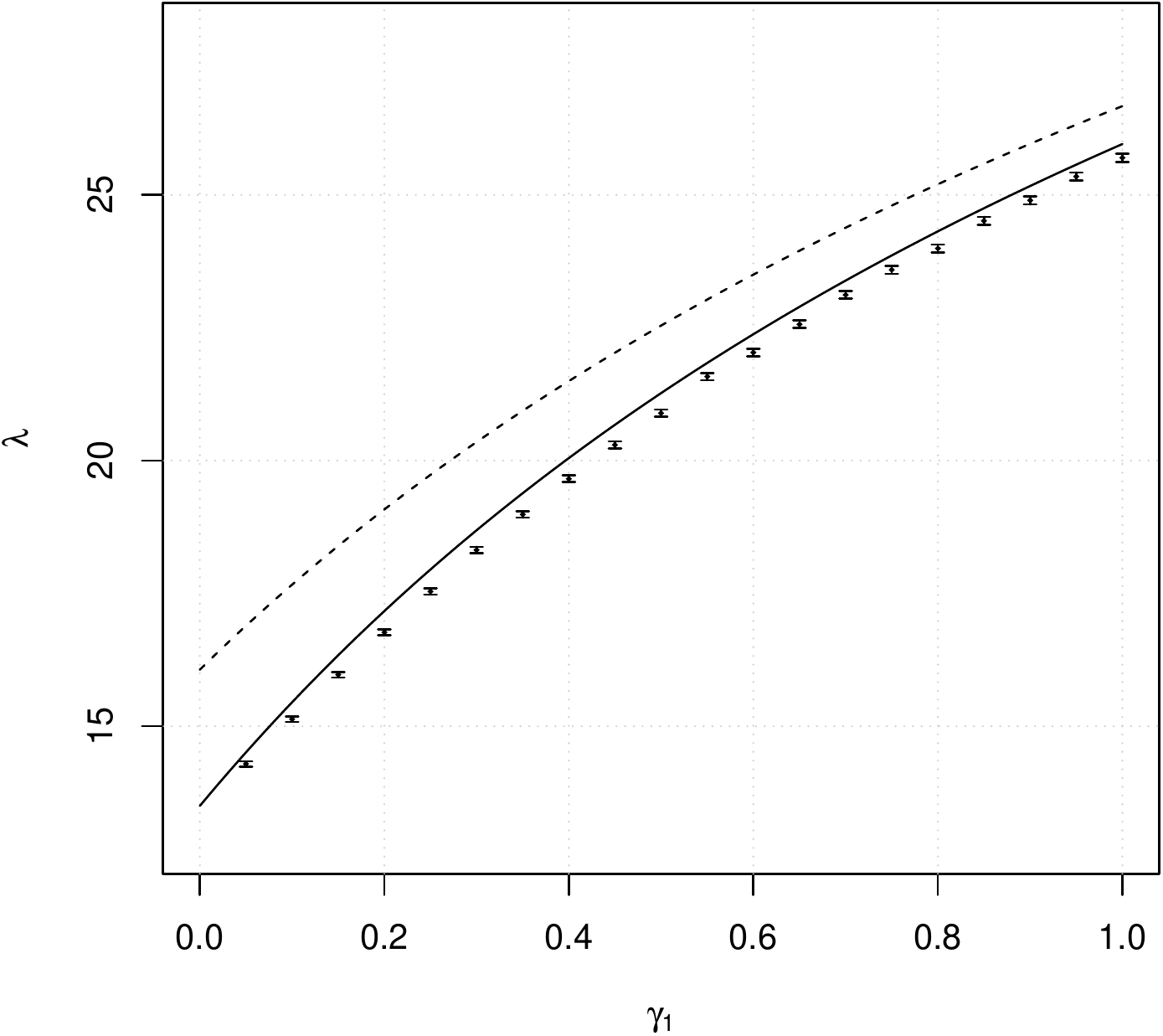}}
\caption{Comparison of the exact intensity (small boxplots obtained by Monte-Carlo method), the Poisson-saddlepoint approximation (dashed line) and the DPP approximation  (solid line) for Diggle-Graton models.  Curves and boxplots are reported in terms of the interaction parameter $\gamma_1 \in [0,1]$. \label{fig:dg} }
\end{figure}

\begin{figure}[h]
\subfigure[PS: $R=(0.05,0.1)$, $\beta=200$, $\gamma_2=0.5$]{\includegraphics[scale=.4]{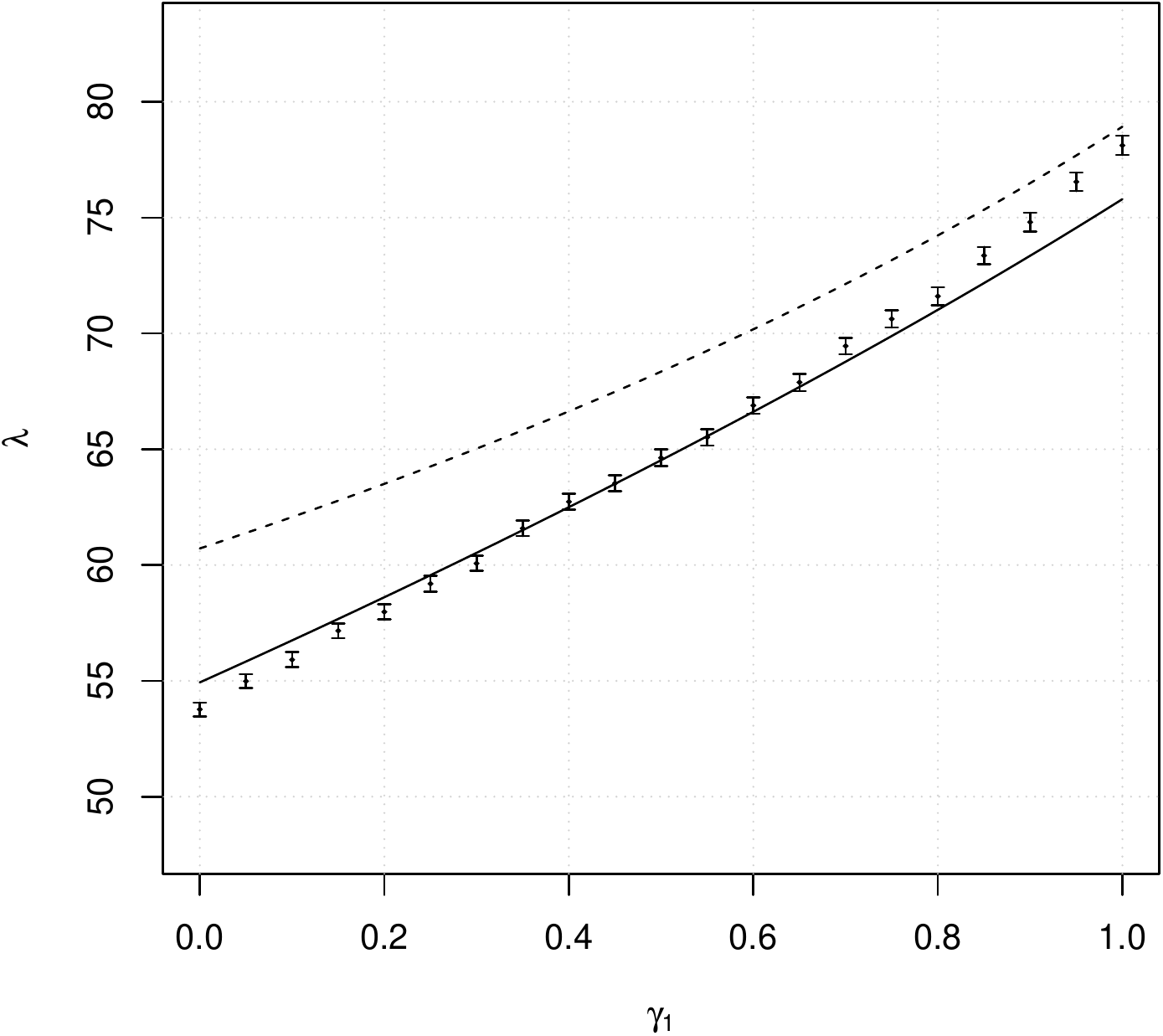}}
\subfigure[PSHC: $\delta=0.025$, $R=(0.05,0.1)$, $\beta=200$, $\gamma_2=0.5$]{\includegraphics[scale=.4]{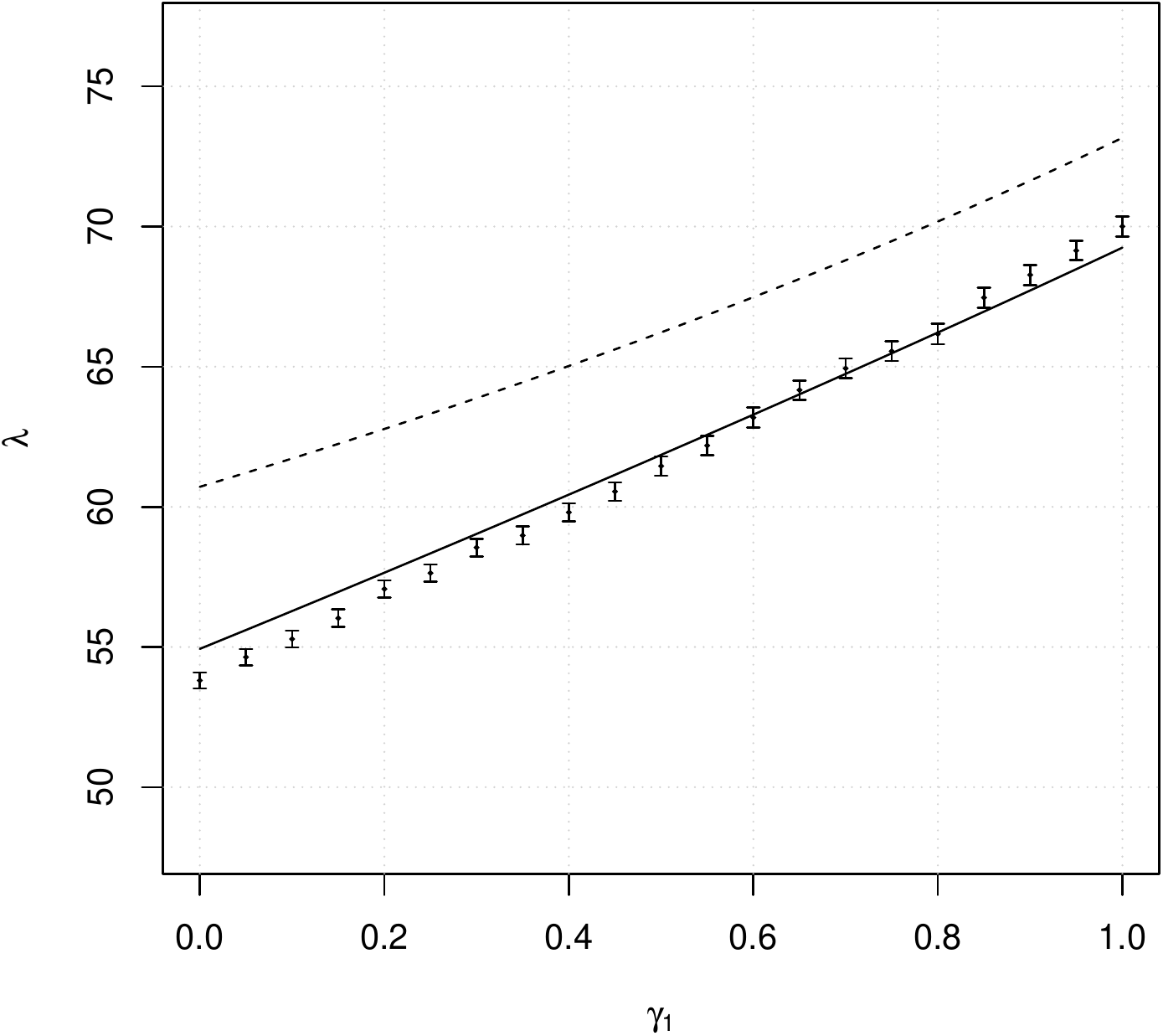}} \\
\subfigure[PS: $R=(0.05,0.1)$, $\beta=200$, $\gamma_2=0$]{\includegraphics[scale=.4]{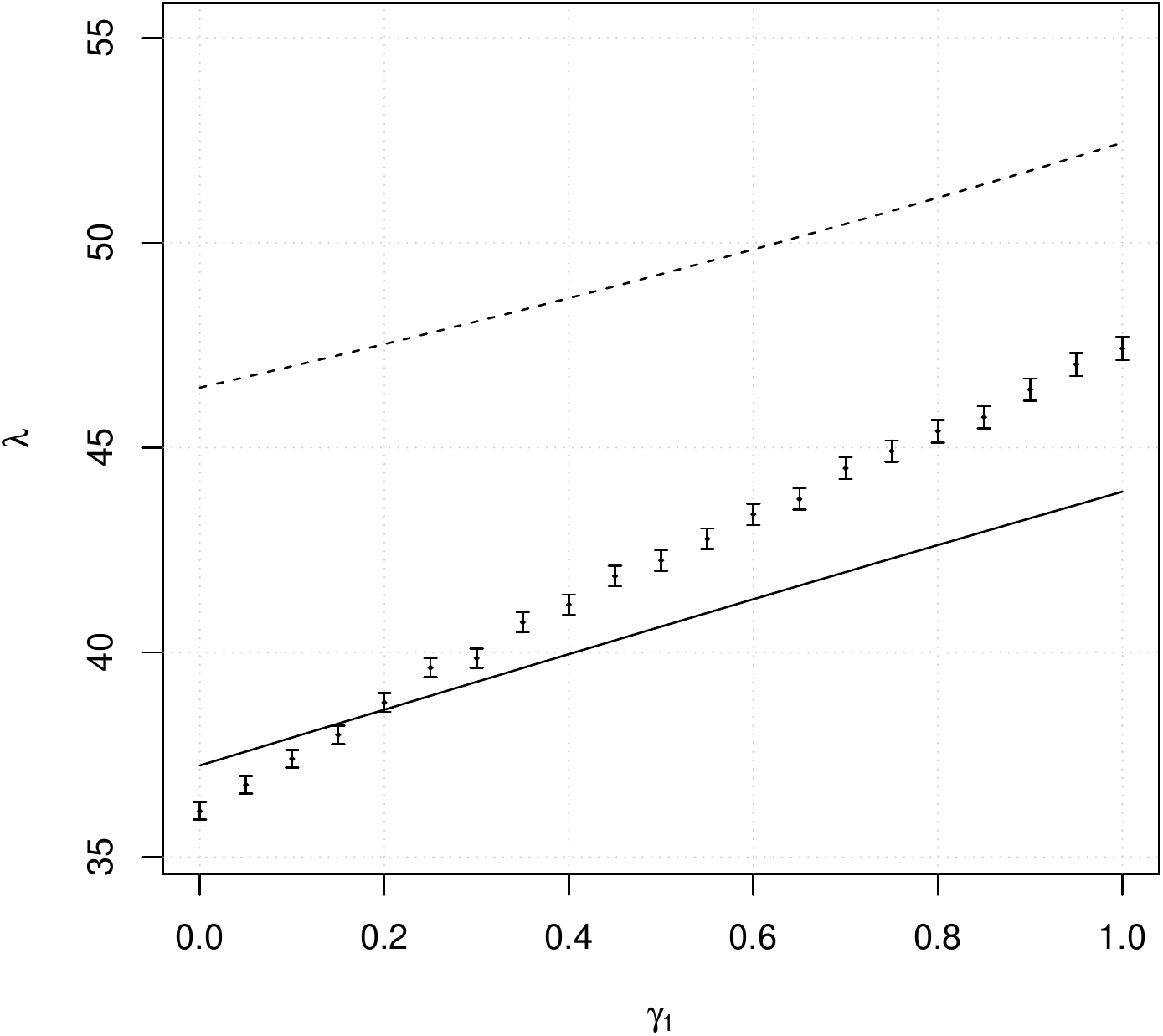}}
\subfigure[PSHC: $\delta=0.025$, $R=(0.05,0.1)$, $\beta=200$, $\gamma_2=0$]{\includegraphics[scale=.4]{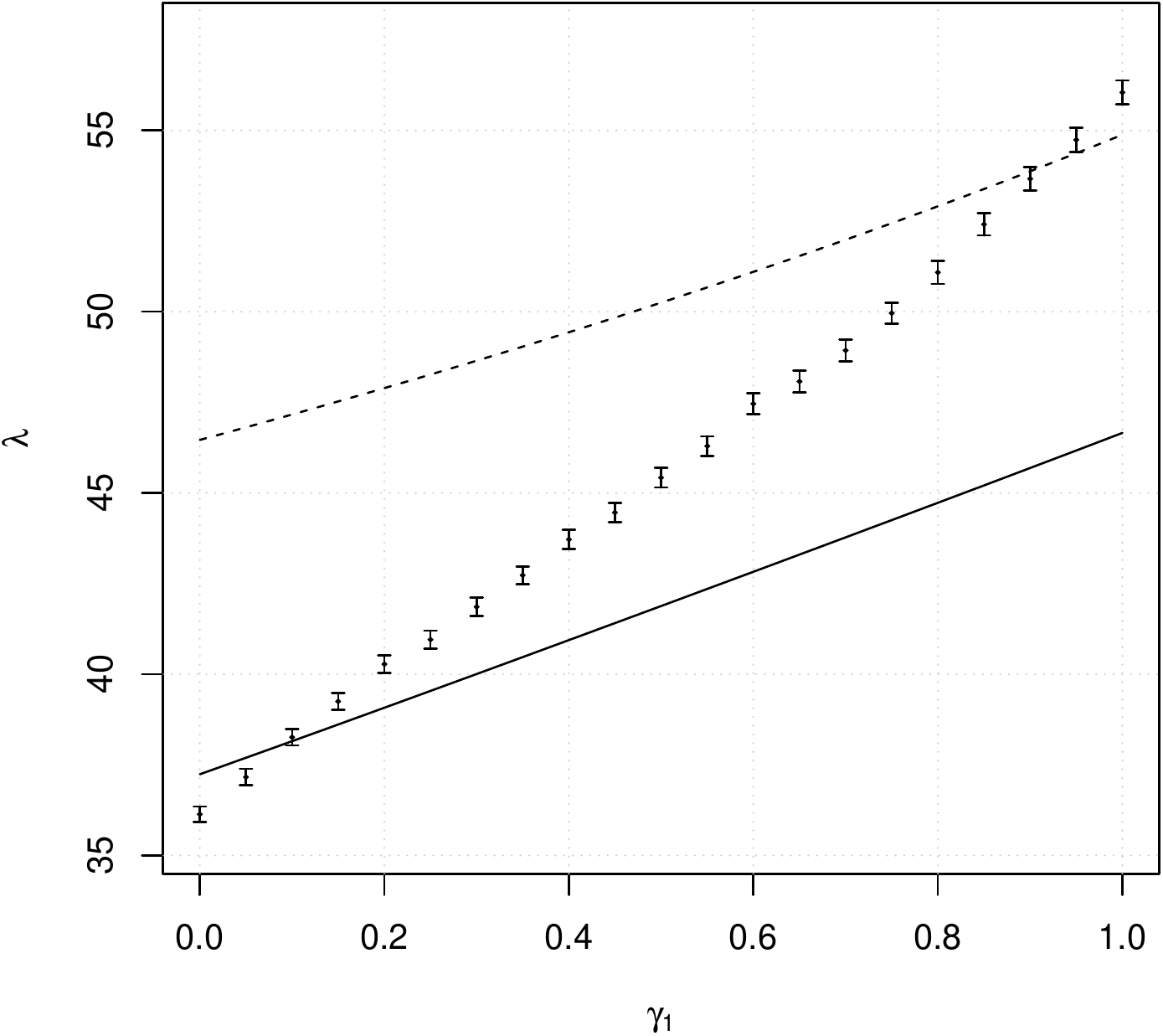}} 
\caption{Comparison of the exact intensity (small boxplots obtained by Monte-Carlo method), the Poisson-saddlepoint approximation (dashed line) and the DPP approximation  (solid line) for piecewise Strauss and piecewise Strauss hard-core models.  Curves and boxplots are reported in terms of the (remaining) interaction parameter $\gamma_1 \in [0,1]$.\label{fig:piecewise} }
\end{figure}

%
%
%
%
\section*{Acknowledgements} The authors are sincerely grateful to Adrian Baddeley and Gopalan Nair for sharing the Monte-carlo replications produced in~\citet{baddeley:nair:12} we used to compare the Poisson-saddlepoint approximation and the DPP approximation  (Figure~\ref{fig:strauss} (a),(b) and (d)). The research of J-F. Coeurjolly is supported by the Natural Sciences and Engineering Research Council of Canada.

\bibliographystyle{plainnat} 
\bibliography{approximation}

\end{document}